\documentclass[10pt]{article}

\usepackage{amsfonts}
\usepackage{amssymb}
\usepackage{amsthm}
\usepackage{amsmath} 
\usepackage{graphicx}
\usepackage{hyperref}
\usepackage{enumerate}
\usepackage{color}
\usepackage[mathscr]{euscript}

\numberwithin{equation}{section}

\newcommand{\ii}{{\rm i}}
\newcommand{\ee}{{\rm e}}
\newcommand{\x}{{\rm x}}

\newcommand{\dd}{{\rm d}}
\newcommand{\beq}{\begin{equation}}
\newcommand{\ene}{\end{equation}}
\newcommand{\ds}{\displaystyle}

\newtheorem{thm}{Theorem}

\newtheorem{prop}[thm]{Proposition}

\newtheorem{rem}[thm]{Remark}

\theoremstyle{definition}
\newtheorem{defn}[thm]{Definition}

\newcommand{\ben}[1]{{\bf {\color{red} BAJ-A: }} \textcolor{red} {``#1"}} 

\begin{document}

\title{Quantum field theory with dynamical boundary conditions and the Casimir effect: Coherent states \thanks{Research partially supported by project  PAPIIT-DGAPA UNAM  IN103918.}}
\author{Benito A. Ju\'arez-Aubry\thanks{Fellow Sistema Nacional de Investigadores. Email: benito.juarez@iimas.unam.mx}\, and Ricardo Weder\thanks{Fellow Sistema Nacional de Investigadores. Email: weder@unam.mx. Home page: http://www.iimas.unam.mx/rweder/rweder.html.}\\
Departamento de F\'isica Matem\'atica, \\
Instituto de Investigaciones en Matem\'aticas Aplicadas y en Sistemas, \\ Universidad Nacional Aut\'onoma de M\'exico,\\Apartado Postal 20-126, Ciudad de M\'exico 01000, M\'exico.}
\date{}

\maketitle

\vspace{.5cm}
\centerline{Abstract}

\noindent {\color{black} We have studied in a previous work} the quantization of a mixed bulk-boundary system describing the coupled dynamics between a \emph{bulk} quantum field confined to a spacetime with finite space slice and  with timelike boundary, and a \emph{boundary} observable defined on the boundary.  Our { bulk} system is a quantum field in a spacetime with timelike boundary and a dynamical boundary condition  -- the boundary observable's equation of motion. Owing to important physical motivations, {\color{black} in such previous work} we have computed the renormalized local state polarization and local Casimir energy for both the bulk quantum field  and the  boundary observable in the ground state and in a Gibbs state at finite, positive temperature. In this work, we introduce an appropriate notion of coherent and thermal coherent states for this mixed bulk-boundary system, and extend our previous study of the renormalized local state polarization and local Casimir energy to coherent and thermal coherent states. \textcolor{black}{We also present numerical results for the integrated Casimir energy and for the Casimir force.}

\section{Introduction}
\label{sec:Intro}

In a previous work \cite{Juarez-Aubry:2020psk} we have studied the quantization in the ground state and at finite, positive temperature of a system describing the coupled dynamics between a bulk scalar field, confined to a region with timelike boundary, and a boundary observable, whose dynamics prescribes \emph{dynamical boundary conditions} for the bulk field, eq. \eqref{Dyn} below. While the system can be seen as a coupled one, it is linear and the quantization can be carried out exactly. One of the main goals of \cite{Juarez-Aubry:2020psk} was to study the local Casimir effect for this system, owing to the fact that dynamical boundary conditions are relevant for the modelling of the realistic setup of experiments verifying the Casimir effect \cite{Fosco:2013wpa, Wilson:2011}. In addition, the study of (classical and quantum) systems with dynamical boundary conditions has several interesting motivations that are discussed in the Introduction of \cite{Juarez-Aubry:2020psk}, both in mathematics (see e.g. \cite{ben99, fulton:1977, mar, menn, shk96, tretter, walter} and references therein) and physics (see e.g. \cite{G.:2015yxa, Barbero:2017kvi, Dappiaggi:2018pju, Fosco:2013wpa, Franca:2019twk, Karabali:2015epa, Martin-Ruiz:2015skg, Inigo, Wilson:2011, Zahn:2015due}).

The purpose of this paper -- which can be seen as a part II to \cite{Juarez-Aubry:2020psk} -- is to extend the analysis of \cite{Juarez-Aubry:2020psk} and study the Casimir effect in coherent states. In order to do so, we shall make use of some of the main results in \cite{Juarez-Aubry:2020psk}, but we are also required to study some aspects of the dynamics of the classical theory from which the quantum theory is defined upon quantization. This is understandable, as our intuition dictates that coherent states in quantum theory are the most ``classical" states, in the sense that they are minimum-uncertainty states \cite{Glauber:1963tx}, and also because they can be seen as ``peaked" around classical configurations.

In this paper, we also explore numerically the integrated Casimir energy and the Casimir force in  particular examples, with the aim of illustrating how this quantity can be obtained from our analytical results with the aid of numerical techniques.

The organization of this paper is as follows. In sec. \ref{sec:Review} we review the mixed bulk-boundary system that we have studied in \cite{Juarez-Aubry:2020psk}, both at the classical and quantum levels. In sec. \ref{sec:Casimir} we introduce coherent states for the quantum theory and briefly discuss how to obtain the local Casimir energy in coherent and thermal coherent states in terms of the local Casimir energy in the ground and finite-temperature states plus a classical contribution related to the energy of the classical solution around which the coherent state is peaked. In sec. \ref{sec:DynBC} we apply the results of sec. \ref{sec:Casimir} to the system presented in sec. \ref{sec:Review}, and obtain the renormalized local state polarizations and local Casimir energies for the bulk quantum field and the boundary observable in coherent states both at zero and finite temperature. For coherent states at zero temperature the main results for the renormalized local state polarization appear in eq. \eqref{DiriPolaCoh} and \eqref{RobinPolaCoh}, and for the local Casimir energy at \eqref{DiriHCoh} and \eqref{RobinHCoh}. For coherent states at positive temperature the main results for the renormalized local state polarization are given in eq. \eqref{DiriPolaCohT} and \eqref{RobinPolaCohT}, and for the local Casimir energy at \eqref{DiriHCohT} and \eqref{RobinHCohT}.  In sec.~\ref{sec:Force} we explore numerically the integrated Casimir energy and the Casimir force in the ground state and for thermal states, and discuss the modification to the  integrated Casimir energy and the Casimir force for coherent states at zero and positive temperature. Sec. \ref{sec:F-loc} contains some structural results of our mixed bulk-boundary theory. We prove that our theory has a unique  causal propagator, that gives the commutation relations of  the quantum fields, and  that it also has  unique advanced and retarded Green's operators.  
Our final remarks appear in sec. \ref{sec:Conc}. We make use of the main results of \cite{Juarez-Aubry:2020psk}, which appear in app. \ref{App:OldResults}.

In this paper we consider the Casimir effect in the interval  $[0,\ell]$ and we assume that our system is well isolated so that we do not have to take into account the exterior of our system. The study of the Casimir effect in a domain in space, without taking into consideration the exterior of the domain has been extensively considered in the literature. See, for example, \cite{Bordag,Fulling:1989nb,Kay:1978zr}.
\section{Quantum field theory with dynamical boundary conditions: a review}
\label{sec:Review}

In \cite{Juarez-Aubry:2020psk} we have studied a mixed bulk-boundary system on the interval $[0, \ell]$ in flat spacetime that describes the coupled dynamics between a bulk scalar field confined to the region inside the interval and a boundary observable. The dynamics of the boundary observable prescribes the boundary conditions for the bulk field at the right end of the interval.

In \cite{Juarez-Aubry:2020psk} a particular interest has been to study the  Casimir effect, owing in good measure to the relevance of these boundary conditions in realistic experimental settings \cite{Fosco:2013wpa, Wilson:2011}. In particular, we have obtained the renormalized local state polarization (the expectation value of the square of the bulk quantum field  and the boundary observable) and the local Casimir energy for the bulk quantum field and the boundary observable in the ground state of the theory and at finite temperature $T = 1/\beta >0$. 

Let us begin by reviewing the studied system and results obtained in \cite{Juarez-Aubry:2020psk}. Consider a scalar field, $\phi: \mathbb R \times [0,\ell] \to \mathbb{R}$, obeying the dynamical equation, with  
$\ell >0.$
\begin{align}
\label{Dyn}
\left\{
                \begin{array}{l}
                  \left[\partial^2_t- \partial^2_z + m^2 + V(z) \right] \phi(t, z) = 0, \, t \in \mathbb{R},  z \in (0,\ell),  \\
                  \cos\alpha \,\phi(t, 0) + \sin\alpha \,\partial_z \phi(t,0) = 0, \, \alpha \in [0, \pi), \\
                  \left[\beta^\prime_1 \partial^2_t -\beta_1 \right] \phi(t, \ell) = - \beta_2 \partial_z \phi(t,\ell) + \beta'_2 \partial_z \partial^2_t \phi(t,\ell),
                \end{array}\right.
\end{align}
where $V(z)$ is a potential, a real valued continuous function defined for  $ z \in [0,\ell]$, $ m^2 \geq 0$ is a mass term, and $\beta_1, \beta_2, \beta_1^\prime, \beta_2^\prime$ are real parameters with the following interpretation: $\beta^\prime_1$ is the square of an inverse velocity, or minus the square of an inverse velocity, $-\beta_1$ is a mass or  a constant potential term and  $\beta_2,$ and $\beta_2^\prime$ are coupling parameters to external sources for the boundary dynamical observable $\phi_\partial(t):= \phi(t,\ell)$. We require that the determinant $\rho :=  \beta_1^\prime\, \beta_2 - \beta_1\, \beta_2' >0$. Subject to initial data at $t = 0$, the system \eqref{Dyn} defines the dynamical interaction between a bulk field, $\phi(t, z)$ with  $0 < z < \ell$, subject to a boundary condition in the Dirichlet or Robin class at $z = 0$, and a boundary observable $\phi_\partial(t):=\phi(t,l)$. If we consider solutions of the form $\phi(t, z)= e^{-i\omega t}\, \varphi(z)$, \eqref{Dyn} becomes a \emph{boundary eigenvalue problem}
 \begin{align}
\label{w.2}
\left\{
                \begin{array}{l}
                  \left[- \partial^2_z + m^2 + V(z) \right] \varphi( z) =  \omega^2\, \varphi(z),  z \in (0,\ell),  \\
                  \cos\alpha \,\varphi(0) + \sin\alpha \,\partial_z \varphi(0) = 0, \, \alpha \in [0, \pi), \\
                  -\left[ \beta_1  \varphi( \ell)  - \beta_2 \partial_z \varphi(\ell)\right]    =  \omega^2 \left[\beta_1^\prime \varphi(\ell)  -\beta'_2 \partial_z  \varphi(\ell)\right],
                \end{array}\right.
\end{align}
with eigenvalue $\lambda = \omega^2$. In this type of problems, the natural Hilbert space of initial data is not $L^2((0,\ell))$, but some extended Hilbert space, $\mathcal H,$ which includes the  boundary observable, and is given by
\beq
\mathcal H:= L^2((0,\ell))\oplus \mathbb C,
\ene
where $\mathbb C $ denotes the complex numbers, with the inner product
\begin{equation}\label{Inner}
(\varphi,\chi)_\mathcal{H} := \int_0^\ell \! \dd z \, \overline{\varphi_1(z)} {\chi_1(z)} +  \rho^{-1} \overline{\varphi_2} {\chi_2},
\end{equation}
for $\varphi = (\varphi_1, \varphi_2)^\top$, $\chi = (\chi_1, \chi_2)^\top \in \mathcal{H}$. We remind the reader that $\rho > 0$. The problem \eqref{Dyn},
\eqref{w.2} can be cast in the form of an abstract Klein-Gordon equation by introducing the linear operator $A$ in $\mathcal{H}$ as in \cite{fulton:1977}, \cite{walter}, 
\begin{align}
\label{opA}
\varphi \in D(A)  \mapsto A\varphi := \begin{pmatrix}
           \left[ - \partial_z^2 + m^2 + V(z) \right] \varphi_1(z)   \\
           -\left[\beta_1 \varphi_1(\ell) - \beta_2 \partial_z  \varphi_1(\ell) \right] \\
         \end{pmatrix},
\end{align}
defined on the domain
\begin{align}
\label{DomA}
D(A) & = \left\{ \varphi =  \begin{pmatrix} 
           \varphi_1   \\
           \varphi_2  \\
         \end{pmatrix} 
\in \mathcal{H}:  \varphi_1, \partial_z \varphi_1  \text{ are absolutely} \,
 \text{continuous in } [0, \ell],  \partial_z^2 \varphi_1 \in L_2((0,\ell)), \right. \nonumber \\
&\left. \cos\alpha \varphi_1(0) + \sin\alpha \partial_z \varphi_1(0) = 0, 	 \varphi_2 =  \beta_1^\prime\varphi_1(\ell) - \beta_2^\prime\partial_z \varphi_1(\ell) \right\}.
\end{align}

It is shown in \cite{fulton:1977} that for $V$ continuous in $[0,l]$, the operator $A$ is densely defined and self-adjoint and has discrete spectrum with multiplicity one eigenvalues that accumulate at $+\infty$. Problem \eqref{w.2} can be written as
\beq\label{w.4}
A \varphi= \omega^2\, \varphi, \qquad   \varphi \in D(A),
\ene
and the dynamical system \eqref{Dyn} can be formulated as the abstract Klein-Gordon equation
\beq\label{w.5}
\partial^2_t \varphi(t,z) + A \varphi(t,z)=0, \qquad \varphi(t,z) \in D(A).
\ene

The following proposition \cite[prop. 1]{Juarez-Aubry:2020psk} gives sufficient conditions for the positivity  of $A$, which allows for canonical quantization to be carried out in the standard way. 
\begin{prop}\label{prop:pos}
Suppose that the determinant $\rho$ is positive,  that $m^2+ V(x) \geq 0, x \in [0,\ell],$  that $ \beta^\prime_2 \neq 0,$  and that either $\alpha =0,$ or $\frac{\pi}{2}  \leq \alpha < \pi.$ Further, assume.
\begin{enumerate}
\item
If $\beta_2^\prime >0,$ then,  $ \beta_1 \geq 0,  \beta_2 < 0,$ and $\beta^\prime_1 < 0.$
\item
If $\beta^\prime_2 < 0,$ then, $ \beta_1 \leq 0,  \beta_2 > 0,$ and $\beta^\prime_1 > 0.$

\end{enumerate}
Under these conditions the eigenvalues of $A$ are all positive, and denoting by $\omega_1^2$ the smallest eigenvalue, $A \geq \omega_1^2 >0.$ 
\end{prop} 

If $ A >0$, in particular  if the conditions of prop. \ref{prop:pos} hold, the solutions to the abstract Klein-Gordon equation \eqref{w.5} take the form
\begin{equation}\label{GenSolClass}
\phi(t) = \cos(t A^{1/2}) f + \frac{\sin (t A^{1/2})}{A^{1/2}} p,
\end{equation}
for initial conditions
\begin{align}\label{w.12}
\left\{
                \begin{array}{l}
                  \phi(0,z) = f(z), \\
                  \partial_t \phi(0,z)  = p(z).
                \end{array}
\right.
\end{align}

If $ f \in D(A)$ and $ p \in D(A^{1/2})$, then \eqref{GenSolClass} defines a strong solution. More generally, for $f,p \in \mathcal H,$ \eqref{GenSolClass} defines a weak solution to \eqref{w.5}. We can write the solutions \eqref{GenSolClass} explicitly in terms of a complete set of orthonormal eigenfunctions of  the self-adjoint operator  $A$, $\{\Psi_{n}\}_{n=1}^\infty$, as
\begin{align}\label{w.13}
\phi(t,z) = \sum_{n = 1}^\infty \, \Psi_{n}(z) \left[(\Psi_{n},f)_\mathcal{H} \cos(\omega_{n} t) + ( \Psi_{n},p)_\mathcal{H} \frac{\sin(\omega_{n} t)}{\omega_{n}} \right],
\end{align}
where  $ \omega_n^2$ is the eigenvalue for the eigenfunction $\Psi_{n}, n=1,2,\dots.$ The series in \eqref{w.13} converges in the norm of $\mathcal H.$ Note that since the operator $A$ is real we can take the eigenfunctions  $\{\Psi_{n}\}_{n=1}^\infty$ with real valued components, and we will always do so below.

The theory above described can be quantized in a canonical sense. We have as the one-particle Hilbert space $\ell^2(\mathcal N)$, which consists of complex-valued, square-summable sequences, 
$\{ \alpha_n \}_{n=1}^\infty$, with the scalar product,
\begin{equation}
\left(  \{ \alpha_n \}_{n=1}^\infty, \{ \beta_n \}_{n=1}^\infty  \right)_{l^2(\mathcal N)}:= \sum_{n=1}^\infty\, \overline{\alpha_n}\, {\beta_n},
\end{equation}
from which the bosonic Fock space \cite{Arai:2018wks} of the theory can be constructed, 
\begin{equation}
\mathscr{H} = \mathbb{C} \oplus_{n = 1}^\infty   ( \otimes_s^n \ell^2(\mathcal N)),
\end{equation} 
where  $\otimes_s^n \ell^2(\mathcal N)$ denotes  the symmetric tensor product of $n$ copies of   $\ell^2(\mathcal N), n=1,\dots$ \cite{Arai:2018wks}. The ground state of the theory, 
\begin{equation}
\Omega_\ell := (1,0,0,\cdots) \in \mathscr{H},
\label{VacuumVector}
\end{equation}
 describes a system at zero temperature. The quantum fields are operators on $\mathscr{H}$ given by 
\begin{align}\label{field}
\hat{\Phi}(t,z) = \sum_{n = 1}^\infty \frac{1}{(2 \omega_n)^{1/2}} \left(\ee^{-\ii \omega_n t} \Psi_n (z)\, {a}_n + \ee^{\ii \omega_n t} \Psi_n (z)\, {a}_n^\dagger \right).
\end{align}
Here, $ a_n$ and $ a_n^\dagger$ are annihilation and creation operators on Fock space, which satisfy the usual canonical commutation relations, 
$[a_n, a_{l}^\dagger] = \delta_{n,l }, [a_n, a_{l}] = 0,[a_n^\dagger, a_{l}^\dagger] =0$ and the $\{\Psi_{n}\}_{n=1}^\infty$ form a complete set of orthonormal eigenfunctions of the self-adjoint operator $A$, defined in \eqref{opA} \eqref{DomA}, with eigenvalues $\omega_n^2$, which can be written as \cite{fulton:1977}
\begin{align}
\Psi_n(z) = \begin{pmatrix}
           \psi_n(z)  \\
           \beta'_1 \psi_n(\ell) - \beta_2'   \partial_z\psi_n(\ell)
         \end{pmatrix}.
\label{Eigenfunctions}
\end{align}

For a Dirichlet boundary condition at $z = 0$, the eigenfunctions \eqref{Eigenfunctions} take the form \eqref{DiriEigen}, while for a Robin boundary condition at $z = 0$ one has \eqref{RobinEigen}.

In \cite{Juarez-Aubry:2020psk} we have studied the Casimir effect for the system that we have just described at zero and at finite positive temperature. Both the renormalized local state polarization and the local Casimir energy can be obtained by using a point-splitting and Hadamard renormalisation prescription. In each case, a  renormalized local state polarization and Casimir energy can be associated to the bulk field and also to the boundary observable. At zero temperature, we have used the two-point Wightman function in the ground state, $\Omega_\ell$, given by\footnote{Eq. \eqref{G1+1} does indeed define a ground state since the interval $[0, \ell]$ is compact. For a theory with non-compact constant-time surfaces a strict ground state, which respects the isometries of the theory, may not be available; this is the situation for Robin boundary conditions in AdS spacetime \cite{Dappiaggi:2016fwc, Pitelli:2019svx}.}
\begin{align}
\langle \Omega_\ell | \hat \Phi(t,z) \, \hat \Phi(t',z') \Omega_\ell \rangle = \sum_{n = 1}^\infty \frac{1}{2\omega_n} \ee^{-\ii \omega_n(t-t')}  \Psi_n (z) \otimes {\Psi_n (z')},
\label{G1+1}
\end{align}
while at positive temperature $T := 1/\beta > 0$ we use the Gibbs state with  the two-point function
\begin{align}
\langle\hat \Phi(t,z)  \hat \Phi(t',z') \rangle_\beta 
&=\sum_{n = 1}^\infty \frac{1}{2\omega_n} \frac{\Psi_n (z)\otimes {{\Psi_n (z')}}}{1-\ee^{-\beta \omega_n}} \left( \ee^{-\ii \omega_n(t-t')} + \ee^{-\beta \omega_n} \ee^{\ii \omega_n(t-t')} \right).
\label{FiniteTMassive}
\end{align}

The two-point Wightman functions \eqref{G1+1} and \eqref{FiniteTMassive} have four tensor-product components with bulk-bulk and boundary-boundary Wightman functions in the diagonal. We have defined in \cite{Juarez-Aubry:2020psk} the bulk renormalized local state polarization and local Casimir energy using the bulk-bulk Wightman function and the boundary counterparts using the boundary-boundary Wightman function. The main results of \cite{Juarez-Aubry:2020psk} appear in Appendix  \ref{App:OldResults}.

\section{Coherent states and the Casimir effect}
\label{sec:Casimir}

It is common folklore that coherent states describe classical-like configurations of quantum fields. Perhaps less known is the fact that for a Klein-Gordon field the stress-energy tensor in a coherent state has the form of a classical term plus a quantum contribution  in globally hyperbolic spacetimes. To the best of our knowledge, in the quantum field theory literature this is mentioned in Note 10 of \cite[sec. 8]{Kay:2013gia} (or Note x in the arXiv version). As we shall see, in our case of interest with timelike boundaries this continues to be the case.

Throughout this section we shall work on the concrete representation of sec. \eqref{sec:Review}, but a large part of the discussion can be generalised directly to the algebraic framework, and holds for quantum field theories in spacetimes with timelike boundaries in general, provided that  they have  causal propagators. 

To define the coherent states we briefly discuss the Segal quantum field operators, following \cite{Arai:2018wks}.   For any  $\{\alpha_n\}_{n=1}^\infty \in \ell^2(\mathcal N),$   we denote by,  $\mathcal A^\dagger(\{ \alpha_n \}_{n=1}^\infty), \mathcal A(\{ \alpha_n \}_{n=1}^\infty),$ respectively, the bosonic creation and annihilation operators over $\mathscr H$ as defined in Section 7 of Chapter 5 of   \cite{Arai:2018wks}. Note that  $\mathcal A^\dagger(\{ \alpha_n \}_{n=1}^\infty),$ and  $\mathcal A(\{ \alpha_n \}_{n=1}^\infty),$ are densely defined and closed and that   they are adjoint of each other, $\mathcal A(\{ \alpha_n \}_{n=1}^\infty)^\star =\mathcal A^\dagger(\{ \alpha_n \}_{n=1}^\infty),$   where for any densely defined operator, $B$ in $\mathscr H$ we denote by $B^\star$ the adjoint of $B.$ Let $\mathcal F_{b, 0}$ be the subset of $\mathscr H$ consisting of all vectors in $ \mathscr H$ with only a finite number of components different from zero, i.e. of all finite particle vectors. It is proven in Proposition 5.5 in page 215 of \cite{Arai:2018wks} that   $\mathcal A^\dagger(\{ \alpha_n \}_{n=1}^\infty)$ and  $\mathcal A(\{ \alpha_n \}_{n=1}^\infty)$  send $\mathcal F_{b, 0}$ into  $\mathcal F_{b, 0},$ and, moreover, that  $\mathcal F_{b, 0}$ is a core for $\mathcal A^\dagger(\{ \alpha_n \}_{n=1}^\infty),$ and  $\mathcal A(\{ \alpha_n \}_{n=1}^\infty).$ Further, it is proven  in Theorem 5.13, in page 218 of  \cite{Arai:2018wks} that the following commutation relations hold,
\begin{align}\label{comm}
\left[ \mathcal  A(\{ \alpha_n \}_{n=1}^\infty),\mathcal  A( \{ \beta_n \}_{n=1}^\infty)^\star \right]=(\{ \alpha_n \}_{n=1}^\infty, \{ \beta_n \}_{n=1}^\infty)_{\ell^2(\mathcal N)},\\
\left[  \mathcal A(\{ \alpha_n \}_{n=1}^\infty), \mathcal  A(\{ \beta_n \}_{n=1}^\infty) \right]= \left[ \mathcal A(\{ \alpha_n \}_{n=1}^\infty)^\star, \mathcal  A(\{ \beta_n \}_{n=1}^\infty)^\star \right]=0,
\end{align}
on $\mathcal F_{b,0}.$

Observe, moreover, that
\beq
\mathcal A(\{ \alpha_n \}_{n=1}^\infty)  \Omega_{\ell}=0, \qquad  \{ \alpha_n \}_{n=1}^\infty \in \ell^2(\mathcal N).
\ene
The creation and annihilation operators that we introduced in  \eqref{field} satisfy,
\beq\label{aa.xx.bb}
a^\dagger_n:= \mathcal  A^\dagger(e_n), \quad a_n:=\mathcal  A(e_n), \qquad n=1,\dots,
\ene
where $e_n$ is the unit vector in $\ell^2(\mathcal N)$ that has component equal one at the place $n$ and all the other components equal to zero.
 
The operator $\mathcal A(\{ \alpha_n \}_{n=1}^\infty) +  \mathcal A(\{ \alpha_n \}_{n=1}^\infty)^\star,$ with domain  $\mathcal F_{b,0},$ is symmetric and densely defined. The Segal quantum field operator is defined as its closure,
\beq\label{segal}
\Phi_S(\{ \alpha_n \}_{n=1}^\infty):= \frac{1}{\sqrt{2}} \overline{\mathcal A(\{ \alpha_n \}_{n=1}^\infty) +\mathcal A(\{ \alpha_n \}_{n=1}^\infty)^\star},
\ene
where the factor $\frac{1}{\sqrt{2}}$ is added for convenience. It is proven in Proposition 5.14 in page 243 of \cite{Arai:2018wks} that $\Phi_S((\{ \alpha_n \}_{n=1}^\infty))$ leaves $\mathcal F_{b,0}$ invariant, and that the following commutation relation holds in $\mathcal F_{b,0},$
\beq\label{segcom}
[\Phi_S(\{ \alpha_n \}_{n=1}^\infty), \Phi_S(\{ \beta_n \}_{n=1}^\infty)]= i \text{\rm Im} (\{ \alpha_n \}_{n=1}^\infty, \{ \beta_n \}_{n=1}^\infty)_{\ell^2(\mathcal N)},
\ene
for $\{ \alpha_n \}_{n=1}^\infty, \{ \beta_n \}_{n=1}^\infty \in \ell^2(\mathcal N).$ Further, in Theorem 5.22 in page 244 of  \cite{Arai:2018wks} it is proven that $\Phi_S(\{ \alpha_n \}_{n=1}^\infty)$ is self-adjoint for all $\{ \alpha_n \}_{n=1}^\infty \in \ell^2(\mathcal N).$

As test functions we take the set $C_{0,R}(\mathbb R, \mathcal H),$ of all continuos functions from $\mathbb R$ into $\mathcal H$ that have compact support in $\mathbb R,$ and take values in functions in $\mathcal H$ that have real-valued components. We take test functions with real-valued components because we  consider a scalar bosonic field with charge zero. For any test function
\beq
 f=\begin{pmatrix} f_1 \\ f_2 
 \end{pmatrix}\in C_{0,R}(\mathbb R, \mathcal H),
\ene
we define the quantum field
\beq\label{qq.ff}
\hat{\Phi}(f)= \Phi_S\left(\mathbf F \int_{\mathbb R}\,dt\,  \frac{e^{i \sqrt{A}\, t}}{A^{1/4}} f(t,\cdot)\right),
\ene 
where  $\mathbf F$ is the unitary operator that diagonalizes $A.$ Namely, $\mathbf F$ is the following unitary operator from $\mathcal H$ onto   $\ell^2(\mathcal N),$
\beq
\mathbf F  \varphi:= \{  (\Psi_n, \varphi)_{\mathcal H} \}_{n=1}^\infty, \qquad \varphi \in \mathcal H.
\ene
Clearly,
\beq
A= \mathbf F^\star \omega_n^2\, \mathbf F.
\ene

By \eqref{segcom} and since $\mathbf F$ is unitary
\beq\label{dd.zz}
[ \hat{\Phi}(f),\hat{\Phi}(g)   ]= -i \int_{\mathbb R^2}\, dt\, dt'\left( \frac{\sin\sqrt{ A}(t-t') }{\sqrt{A}}f(t,\cdot), g(t',\cdot) \right)_{\mathcal H}, \qquad f, g  \in C_{0,R}(\mathbb R, \mathcal H).
\ene
We introduce the following notation,
\beq\label{aaa.bbb}
E(t-t')= -\frac{\sin \sqrt{A}(t-t')}{\sqrt{A}},
\ene
and we define the following integral operator in  with domain $C_{0, R}(\mathbb R, \mathcal H),$ and range contained in the set  of all continuos functions from $\mathbb R$ into $\mathcal H$ that  take values in functions in $\mathcal H$ that have real-valued components, which we denote by $ C_R(  \mathbb R,  \mathcal H)$,
\beq\label{daaa.bbb}
\mathsf E \varphi= \int_{\mathbb R}\, dt'\, E(t-t')\, \varphi(t',\cdot).
\ene
Hence, \eqref{dd.zz}  can be written as follows,
\beq\label{dd.zz.ww}
[ \hat{\Phi}(f),\hat{\Phi}(g)   ]= i \int_{\mathbb R}\, dt\,\left((\mathsf  E f)(t,\cdot), g(t,\cdot)\right)_{\mathcal H}, \qquad f, g \in C_{0,R}(\mathbb R, \mathcal H).
\ene

Note that, by  \eqref{field}, \eqref{aa.xx.bb},  and since  the annihilation operator $ \mathcal A(\{\alpha_n\}_{n=1}^\infty)$ is antilinear in   $\{\alpha_n\}_{n=1}^\infty,$
\beq\label{ee.qq}
\hat{\Phi}(f)= \int_{\mathbb R}\, dt \, \int_0^\ell \, dz\, \hat{\Phi}(t,z) f(t,z).
\ene 
We define a coherent state in Fock space, $F,$ by 
\beq\label{a.1}
F := \ee^{\ii \hat{\Phi}(f)} \Omega_\ell, \qquad f \in  C_{0,R}(\mathbb R, \mathcal H).
\ene
Observe that the  Weyl operator $\ee^{\ii \hat \Phi(f)}$ is a unitary operator on Fock space since the quantum  field operator $\hat{\Phi}(f)$ is self-adjoint. 

By Theorem 5.28 in page 263 of \cite{Arai:2018wks}  and since $\mathbf F$ is unitary, we have that,
\beq\label{zzz.1}
\ee^{-i \hat{\Phi}(f)} \hat{\Phi}(g) \ee^{i \hat{\Phi}(f)}=  \hat{\Phi}(g) + \text{\rm Im} \left(\int_{\mathbb R}\,dt'\,  \frac{e^{i \sqrt{A}\, t'}}{A^{1/4}} f(t',\cdot), \int_{\mathbb R} \,dt\, \frac{e^{i \sqrt{A}\, t}}{A^{1/4}} g(t,\cdot)\right)_{\mathcal H}.
\ene
Let us denote,
\beq\label{www.333}
\phi_f(x)=\left(Ef \right)(x):= \left(\int_{\mathbb R}\,dt'  \,E(t-t')\, f(t', \cdot)\right)(z),
\ene
where $ x:=(t,z)$ and $f \in C_{0,R}(\mathbb R, \mathcal H).$  The function  $\phi_f(x)$  is a weak solution to  \eqref{w.5}. 
   By  \eqref{zzz.1} and  \eqref{www.333},
\beq\label{qq.aa}
\ee^{-i \hat{\Phi}(f)} \hat{\Phi}(g) \ee^{i \hat{\Phi}(f)}=  \hat{\Phi}(g) - \int_{\mathbb R}\, dt\, \int_0^\ell\, dz  \phi_f(t,z) g(t,z).
\ene
Using  \eqref{qq.aa} we verify that  the $n$-point  Wightman function in the coherent state $F$ satisfies,
\begin{align}
\langle F | \hat \Phi(\x_1) \cdots \hat \Phi(\x_n) F \rangle = \langle \Omega_\ell | \prod_{i = 1}^n \left( \hat \Phi(\x_i) - \phi_f(\x_i) 1\!\!1 \right) \Omega_\ell \rangle,
\label{Coh0}
\end{align}
where $x_j:=(t_j,z_j).$
Thermal coherent states can be defined analogously,
\begin{align}
\langle \hat \Phi(\x_1) \cdots \hat \Phi(\x_n) \rangle_\beta^F := \langle  \ee^{-\ii \hat \Phi(f)} \hat \Phi(\x_1) \cdots \hat \Phi(\x_n) \ee^{\ii \hat \Phi(f)} \rangle_\beta = \langle  \prod_{i = 1}^n \left( \hat \Phi(\x_i) -(\phi_f)(\x_i) 1\!\!1 \right) \rangle_\beta.
\label{CohT}
\end{align}

It follows from the two-point functions defined by eq. \eqref{Coh0} and \eqref{CohT}  that the renormalized stress-energy tensor in a coherent or thermal coherent state can be obtained in terms of the ground state or finite-temperature renormalized stress-energy tensor via a point-splitting prescription. Let us discuss very briefly the essentials of stress-energy renormalisation for the Klein-Gordon theory to see that this is the case. 

Physically relevant states for the Klein-Gordon theory, such as the ground state and thermal states, satisfy the so-called Hadamard condition for the short-distance limit of their two-point function. In particular, a Klein-Gordon state, $\eta ,$ is Hadamard if for any two points, $\x, \x'$, in a convex normal neighbourhood of spacetime, $N$, it holds for the two-point function that
\begin{align}
\eta(\hat \Phi(\x) \hat \Phi(\x'))-H_{\lambda}(\x, \x') \in C^\infty(N \times N),
\end{align}
where $H_{\lambda}$ is the Hadamard bi-distribution of the appropriate dimension of spacetime, see \cite{Decanini:2005gt}. The properties, construction and ambiguities in defining $H_\lambda$ are well known, and we do not intend to discuss them throughly. Instead, we refer to the reader to \cite{Dappiaggi2017, Decanini:2005gt, Kay:1988mu, Wald:1995yp} for details, and present explicitly the Hadamard bi-distribution in $(1+1)$-dimensional flat spacetime that is relevant in our analysis of sec. \ref{sec:DynBC}. 

In $1+1$ Minkowski spacetime for the Klein-Gordon field obeying $(\frac{\partial^2}{\partial t^2}- \frac{\partial^2}{\partial x^2}  + m^2)\phi = 0$, we fix the Hadamard bi-distribution in terms of coordinates $x = (t,z)$, $x' = (t',z')$ to
\begin{align}
 H_{\rm M}&((t,z),(t',z'))  :=  -\frac{1}{4 \pi} \left\{ 2 \gamma + \ln \left[ \frac{m^2}{2} \sigma((t,z),(t',z')) \right] \right.  + \frac{m^2}{2} \sigma((t,z),(t',z'))  \left[  \ln \left(\frac{m^2}{2} \sigma ((t,z),(t',z')) \right) + 2\gamma- 2   \right]  \nonumber \\
& \left. + \frac{m^4}{16 \pi} \sigma^2((t,z),(t',z')) \left[ \ln \left(\frac{m^2}{2} \sigma((t,z),(t',z')) \right)+2 \gamma -3 \right] \right\}  + O\left(\sigma^3((t,z),(t',z') \ln \left( \sigma((t,z),(t',z') \right) \right),
\label{HM}
\end{align}
with $2 \sigma((t,z),(t',z')) = -(t-t')^2+(z-z')^2$ and $\gamma$ the Euler number. This choice satisfies Wald's fourth axiom for the renormalized stress-energy tensor \cite{Wald:1995yp} and guarantees that the renormalized stress-energy tensor in Minkowski spacetime be zero, and was used in our previous paper \cite{Juarez-Aubry:2020psk}. The subscript $\rm M$ on the left-hand side of  \eqref{HM} places emphasis on this fact.

Using eq. \eqref{Coh0} and \eqref{CohT} it can be seen that the two-point functions singularities in coherent and thermal coherent states are also of Hadamard form and, provided that the classical solution around which the coherent state is built be differentiable, the expectation value of the stress-energy tensor (and of the Hamiltonian in particular) can be defined in these states. This will turn out to be the case, as we shall see in prop. \ref{prop:RegDiri} and \ref{prop:RegRobin}.

The renormalized state polarization in the Hadamard state $\eta$ is defined as follows,
\beq\label{ww.zz}
\eta((\hat{\Phi}_{\rm ren}(x))^2):= \lim_{\x' \to \x}  \left( \eta \left( \hat \Phi(\x) \hat \Phi(\x') \right) - H_{\rm M}(\x, \x') \right).
\ene 

The renormalized stress-energy tensor of the Klein-Gordon field in the Hadamard state, $\eta,$ with two-point function $\eta\left( \hat \Phi(\x) \hat \Phi(\x') \right)$ is defined by a point-splitting prescription
\begin{align}
\eta\left( \hat T^{\rm ren}_{ab}(\x) \right) & := \lim_{\x' \to \x} \mathscr{T}_{ab} \left( \eta \left( \hat \Phi(\x) \hat \Phi(\x') \right) - H_{\rm M}(\x, \x') \right) + \Theta_{ab}(\x),
\label{Tren}
\end{align}
where $\Theta_{ab}$ is any covariantly conserved, symmetric tensor built out of the metric and its derivatives -- it is an ambiguity -- and $\mathscr{T}_{ab}$ is a point-splitted differential operator version of the classical stress-energy tensor, defined in such a way that the classical stress-
energy tensor is given by $T_{ab}^{\phi}(\x) = \lim_{\x' \to \x} \mathscr{T}_{ab} \phi(\x) \phi(\x')$, where $\phi$ is a classical configuration. See \cite[eq. 5.85]{KhavkineMoretti} for the detailed form of $\mathscr{T}_{ab}$.

Suppose now that we consider the coherent and thermal coherent states defined in \eqref{Coh0} and \eqref{CohT}, respectively. We have that
\begin{subequations}
\label{TabCoherent}
\begin{align}
\langle F | \hat T^{\rm ren}_{ab}(\x) F \rangle & = \lim_{\x' \to \x} \mathscr{T}_{ab} \left(\langle \Omega_{\ell} | \hat \Phi(\x) \hat \Phi(\x') \Omega_{\ell} \rangle +\phi_f(\x) \phi_f(\x') - H_{\rm M}(\x, \x') \right) + \Theta_{ab}(\x) = \langle \Omega_{\ell} | \hat T^{\rm ren}_{ab}(\x) \Omega_{\ell} \rangle + T^{\phi_f}_{ab}(\x), \\
\langle \hat T^{\rm ren}_{ab}(\x) \rangle_\beta^F & = \lim_{\x' \to \x} \mathscr{T}_{ab} \left( \langle \hat \Phi(\x) \hat \Phi(\x') \rangle_\beta + \phi_f(\x) \phi_f(\x') - H_{\rm M}(\x, \x') \right) + \Theta_{ab}(\x) = \langle \hat T^{\rm ren}_{ab}(\x) \rangle_\beta + T^{\phi_f}_{ab}(\x),  
\end{align}
\end{subequations}
where we used that the for the ground state and the Gibbs state the one-point function is zero, and where   $T^{\phi_f}_{ab}$ is the classical stress-energy tensor around the classical configuration $\phi_f = E f$. Furthermore, we have taken the ambiguity $\Theta_{ab}$ equal to zero.  Eq. \eqref{TabCoherent} are fully analogous to the expressions in globally hyperbolic spacetimes. 

The local Casimir energy (or energy density) can be read off from the time-time component (in the sense of the distinguished time in static spacetimes) of eq. \eqref{TabCoherent}. In $1+1$ dimensions in flat spacetime we have
\begin{subequations}
\label{HCoherent}
\begin{align}
\langle F | \hat H_{\rm ren}(\x) F \rangle & =  \langle \Omega_{\ell} | \hat H_{\rm ren}(\x) \Omega_\ell \rangle + \frac{1}{2} \left( (\partial_t \phi_f(\x))^2 + (\partial_z \phi_f(\x))^2 + m^2 (\phi_f(\x))^2 \right), \label{Hground} \\
\langle \hat H_{\rm ren}(\x) \rangle_\beta^F & = \langle \hat H_{\rm ren}(\x) \rangle_\beta + \frac{1}{2} \left( (\partial_t \phi_f(\x))^2 + (\partial_z \phi_f(\x))^2 + m^2 (\phi_f(\x))^2 \right). \label{HKMS}  
\end{align}
\end{subequations}

\section{The Casimir effect for quantum field theory with dynamical boundary conditions in coherent states}
\label{sec:DynBC}

Recall that we  assume that  $ A >0.$  In particular, this is true if the condition of prop. \ref{prop:pos} are met. Henceforth, we shall assume  that the potential $V(z) = 0$.

The details of the quantization of system \eqref{Dyn} or equivalently \eqref{w.5} with $V(z) = 0$ and under the assumption that $A>0$  are presented in detail in our work \cite{Juarez-Aubry:2020psk} at zero and finite, positive temperature in the cases in which there is a Dirichlet ($\alpha = 0$) or Robin ($\alpha \neq 0$) boundary condition at $z = 0$, and the bulk and boundary local Casimir energies and renormalized local state polarizations are obtained in each case. For the purposes of this paper, it suffices for us to quote the results obtained in \cite{Juarez-Aubry:2020psk} for the renormalized local state polarization and local Casimir energy at zero and finite, positive temperature, as this is all we need in order to compute the renormalized local state polarization and local Casimir energy for coherent and thermal coherent states respectively, in view of eq. \eqref{HCoherent}. These results appear in Appendix \ref{App:OldResults}.

We now obtain the Casimir effect in coherent and thermal coherent, in the ground state and at positive, finite temperature, both in the case of a Dirichlet boundary condition  at $z=0$ in sec. \ref{subsec:Diri} and for a Robin boundary condition at $z = 0$ in sec. \ref{subsec:Robin}.  Note that the vacuum state $\Omega_\ell$ does not depend on the boundary condition. Anyhow, we  use below the notation $\Omega_\ell^{ ({\rm D })},$   $\Omega_\ell^{ ({\rm R })}$ to emphasize that we consider, respectively,  the  Dirichlet or the Robin boundary condition at $z=0.$

\subsection{Dirichlet boundary condition at $z = 0$}
\label{subsec:Diri}

For Dirichlet boundary condition at $z = 0$, the ground state two-point Wightman   function is
\begin{align}
\label{WightmanDiri}
\langle \Omega_\ell^{({\rm D})} | \hat \Phi(t,x) \hat \Phi(t',x') \Omega_\ell^{({\rm D})} \rangle = \sum_{n = 1}^\infty \frac{1}{2\omega_n^{\rm D}} \ee^{-\ii \omega_n^{\rm D}(t-t')} \Psi_n^{({\rm D})} (z) \otimes \Psi_n^{({\rm D})} (z'),
\end{align}
where $\Omega_\ell^{\rm (D)}$ is the ground state in Fock space, and where $(\omega_n^{\rm D})^2 = (s_n^{\rm D})^2 + m^2$ are the eigenvalues of the operator $A$ \eqref{opA} \eqref{DomA} with $\alpha = 0$ (see \cite[App. B]{Juarez-Aubry:2020psk} for asymptotic estimates of the eigenvalues at large $n$) and with eigenfunctions (normalised with respect to the inner product \eqref{Inner}) given by 
\begin{align}
\Psi_n^{(\rm D)}(z) & = \begin{pmatrix}
				\psi_n^{({\rm D})}(z) \\ 
				\psi_n^{\partial \, ({\rm D})}
\end{pmatrix} 
 = \mathcal{N}_n^{\rm D} 
	\begin{pmatrix}
          -\sin(s_n^{\rm D} z)  \\
           -\beta_1' \sin(\ell s_n^{\rm D}) + \beta_2' s_n^{\rm D} \cos(\ell s_n^{\rm D})
    \end{pmatrix}, \nonumber \\
\mathcal{N}_n^{\rm D} & = \left(\frac{\ell}{2} + \frac{2 \rho (s_n^{\rm D})^2 - (\beta_1' (s_n^{\rm D})^2 + (\beta_1+ \beta_1' m^2))(\beta_2' (s_n^{\rm D})^2 + (\beta_2+ \beta_2' m^2))}{2 \left[ (s_n^{\rm D})^2(\beta_2' (s_n^{\rm D})^2 + (\beta_2+ \beta_2' m^2))^2 + (\beta_1' (s_n^{\rm D})^2 + (\beta_1+ \beta_1' m^2))^2 \right]} \right)^{-1/2}.
\label{DiriEigen}
\end{align}

At finite temperature $T = 1/\beta > 0$, the  two-point Wightman  function is given by
\begin{align}\label{we.100}
\langle\hat \Phi(t,z)  \hat \Phi(t',z') \rangle_{\beta,{\rm D}} &=  \langle \Omega_\ell^{\rm (D)} | \hat \Phi(t,z) \hat \Phi(t',z') \Omega_\ell^{\rm (D)} \rangle + \sum_{n = 1}^\infty \,\frac{1}{2 \omega_n^{\rm (D)}} \,\Psi_n^{\rm (D)} (z)  \otimes \Psi_n^{\rm (D)} (z')\, \,\frac{1}{\ee^{\beta \omega_j^{\rm D}} - 1}\,\left(\ee^{-\ii \omega_n^{\rm D}(t-t')} + \ee^{\ii \omega_n^{\rm D}(t-t')} 
\right).
\end{align}

Both at zero and finite temperature, the bulk renormalized local state polarization and local Casimir energy are computed making use of the bulk-bulk component of \eqref{WightmanDiri} or \eqref{we.100}. For example, for the two-point Wightman  function \eqref{WightmanDiri}, this is 
\begin{align}
\langle \Omega_\ell^{({\rm D})} | \hat \Phi^{\rm B}(t,z) \hat \Phi^{\rm B}(t',z') \Omega_\ell^{({\rm D})} \rangle & = \sum_{n = 1}^\infty \frac{1}{2\omega_n^{\rm D}} \ee^{-\ii \omega_n^{\rm D}(t-t')} \psi_n^{({\rm D})} (z) \psi_n^{({\rm D})} (z'),
\end{align}
and similarly for the two-point Wightman function at finite temperature \eqref{we.100}. For the boundary state polarization and Casimir energy at zero or finite temperature, one makes use the boundary-boundary component of \eqref{WightmanDiri} or \eqref{we.100}. For example, for the  two-point Wigthman function \eqref{WightmanDiri}, this is
\begin{align}
\langle \Omega_\ell^{({\rm D})} | \hat \Phi^\partial(t) \hat \Phi^\partial(t') \Omega_\ell^{({\rm D})} \rangle & = \sum_{n = 1}^\infty \frac{1}{2\omega_n^{\rm D}} \ee^{-\ii \omega_n^{\rm D}(t-t')} \psi_n^{\partial \, ({\rm D})} \psi_n^{\partial \, ({\rm D})},
\end{align}
and similarly for the two-point Wightman function at finite temperature \eqref{we.100}. The commutator in the Dirichlet case  is given by, 
\beq\label{causalD}
\mathsf E_{\rm D} \varphi (t,z):= \int_{\mathbb R}\, dt'\, dz'\, E_{\rm D} \left( (t,z), (t',z') \right)\, \varphi(t',z'),
\ene
with integral kernel,

\begin{align}
E_{\rm D} \left( (t,z), (t',z') \right) = -\sum_{n = 1}^\infty \frac{1}{\omega_n^{\rm D}} \sin\left(\omega_n^{\rm D}(t-t')\right)  |\Psi_n^{({\rm D})} (z) \rangle\,\langle \Psi_n^{({\rm D})} (z')|.
\label{AdvRetDiri}
\end{align}

\subsubsection{  Renormalized local state polarization in coherent and thermal coherent states}

Let $f =  \begin{pmatrix} 
           f_1   \\
           f_2  \\
         \end{pmatrix} \in C_{0, R}^\infty(\mathbb{R}, \mathcal H)$, then $F^{\rm (D)} = \ee^{\ii \hat \Phi(f)} \Omega_\ell^{\rm (D)}$ defines a coherent state in the case of a Dirichlet boundary condition at $z = 0$ ($\alpha = 0$) around the classical solution $\phi_f^{\rm D} = E_{\rm D} f$, which can be written in detail as 
\begin{align}
\phi_f^{\rm D}(t,z) = -\int_{\mathbb{R}} \! \dd t' \, \sum_{n = 1}^\infty \frac{1}{\omega_n^{\rm D}} \sin\left(\omega_n^{\rm D}(t-t')\right) \begin{pmatrix}
				\psi_n^{({\rm D})}(z) \left( \int_0^\ell \! \dd z' \, \psi_n^{({\rm D})}(z') f_1(t',z') + \rho^{-1} \psi_n^{\partial \, ({\rm D})} f_2(t') \right) \\ 
				\psi_n^{\partial \, ({\rm D})} \left( \int_0^\ell \! \dd z' \, \psi_n^{({\rm D})}(z') f_1(t',z') + \rho^{-1} \psi_n^{\partial \, ({\rm D})} f_2(t') \right)
\end{pmatrix}. 
\label{DiriClass}
\end{align}

The bulk and boundary renormalized local state polarizations in the coherent state $F^{\rm (D)}$ are given respectively by (see \eqref{Coh0} and
\eqref{ww.zz})
\begin{subequations}
\label{DiriPolaCoh}
\begin{align}
\langle F^{(\rm D)} | ( \hat \Phi^{\rm B}_{\rm ren} )^2 (t,z) F^{(\rm D)} \rangle & = \langle \Omega^{(\rm D)} | ( \hat \Phi^{\rm B}_{\rm ren} )^2 (t,z) \Omega^{(\rm D)} \rangle + \left((\phi_{f}^{\rm D})_1(t,z)\right)^2, \\
\langle F^{(\rm D)} | ( \hat \Phi^{\partial}_{\rm ren} )^2 (t) F^{(\rm D)} \rangle & = \langle \Omega^{(\rm D)} | ( \hat \Phi^{\partial}_{\rm ren} )^2 (t) \Omega^{(\rm D)} \rangle + \left((\phi_{f}^{\rm D})_2(t)\right)^2,
\end{align}
\end{subequations}
where we used that for the ground state the one-point function is zero. The bulk and boundary renormalized local  state polarizations  in the ground state are given in app \ref{App:OldResults}, eq. \eqref{DirichletVacuumBulk} and \eqref{DirichletVacuumBound} respectively.  Here, $(\phi_{f}^{\rm D})_1$ makes reference to the first component of $\phi_{f}^{\rm D}$ and $(\phi_{f}^{\rm D})_2$ to the second one.

A thermal coherent state around the solution $\phi_f^{\rm D}$ with a Dirichlet boundary condition at $z = 0$ has bulk and boundary two-point Whigtman functions given respectively by (see \eqref{CohT})
\begin{subequations}
\begin{align}
\langle \hat \Phi^{\rm B}(t,z) \hat \Phi^{\rm B}(t',z') \rangle_{\beta,{\rm D}}^F & = \langle \hat \Phi^{\rm B}(t,z) \hat \Phi^{\rm B}(t',z') \rangle_{\beta, {\rm D}} +( \phi_f^{\rm D})_1(t,z)(\phi_f^{\rm D})_1(t',z'), \\
\langle \hat \Phi^{\partial}(t) \hat \Phi^{\partial}(t') \rangle_{\beta,{\rm D}}^F & = \langle \hat \Phi^{\partial}(t) \hat \Phi^{\partial}(t') \rangle_{\beta, {\rm D}} + (\phi_f^{\rm D})_2(t)(\phi_f^{\rm D})_2(t'),
\end{align}
where we used that for the Gibbs state the one-point function is zero. The two-point Wightman function $\langle \hat \Phi(\x) \hat \Phi(\x') \rangle_{\beta, {\rm D}}$ is given by \eqref{we.100}.
\end{subequations}

The bulk and boundary renormalized local state polarizations in the thermal coherent state $\langle \cdot \rangle_{\beta,{\rm D}}^F$ are given respectively by (see  \eqref{ww.zz})
\begin{subequations}
\label{DiriPolaCohT}
\begin{align}
\langle ( \hat \Phi^{\rm B}_{\rm ren})^2(t,z) \rangle_{\beta,{\rm D}}^F& = \langle ( \hat \Phi^{\rm B}_{\rm ren})^2(t,z) \rangle_{\beta,{\rm D}} + \left((\phi_{f}^{\rm D})_1(t,z)\right)^2, \\
\langle ( \hat \Phi^{\partial}_{\rm ren})^2(t) \rangle_{\beta,{\rm D}}^F & = \langle ( \hat \Phi^{\partial}_{\rm ren})^2(t) \rangle_{\beta,{\rm D}} + \left((\phi_{f}^{\rm D})_2(t)\right)^2,
\end{align}
\end{subequations}
where the bulk and boundary renormalized local state polarization at finite temperature are given in app. \ref{App:OldResults}, eq. \eqref{TBulkPolaD} and \eqref{TBoundPolaD}
 respectively.

\subsubsection{Local Casimir energy in coherent and thermal coherent states}

We now address the local Casimir energy with a Dirichlet boundary condition at $z = 0$ for coherent and thermal coherent states around the classical solution $\phi_f^{\rm D}$ \eqref{DiriClass}. 

To define the renormalized local Casimir energy we need to differentiate the solution $\phi_f,$ cf. formulae \eqref{HCoherent}. For this purpose we need to assume more regularity on $f.$ To  make precise the regularity that we need we introduce the graph norm of $\sqrt{A}$,
\beq
\|  \varphi \|_{\sqrt{A}}:= \|\varphi\|_\mathcal H+ \| \sqrt{A}\varphi\|_{\mathcal H}, \qquad   \varphi \in D(\sqrt{A}).
\ene
The domain of $\sqrt{A},$ that we denote by $ D(\sqrt{A}),$ endowed with the graph norm is a Hilbert space.   The space of test functions that we need to differentiate  the $\phi_f$ is  
$C_{0,R}(\mathbb R, D(\sqrt{A})),$  that consists of all continuos functions from $\mathbb R$ into  $D(\sqrt{A})$ endowed with the graph norm, that have compact support in $\mathbb R,$ and take values in functions in $ D(\sqrt{A})$ with   real-valued components.  

\begin{prop}
Let $\phi_f^{\rm D} = E_{\rm D} f$ be as above a solution to eq. \eqref{w.5} with
 \beq
 f =  \begin{pmatrix} 
           f_1   \\
           f_2  \\
         \end{pmatrix} \in C_{0,R}(\mathbb R,  D(\sqrt{A})).
 \ene
   
Then, $\phi_f^{\rm D}(t,z)$ is continuously differentiable in $t$ and $z,$ and the derivatives are uniformly bounded for $ t \in \mathbb R$ and $ z \in [0,\ell].$
\label{prop:RegDiri}
\end{prop}
\begin{proof}
Since $f \in D(\sqrt{A})$, equation \eqref{DiriClass} can be written as follows,
\begin{align}
\phi_f^{\rm D}(t,z) = -\int_{\mathbb{R}} \! \dd t' \, \sum_{n = 1}^\infty \frac{1}{(\omega_n^{\rm D})^2} \sin\left(\omega_n^{\rm D}(t-t')\right) \begin{pmatrix}
				\psi_n^{({\rm D})}(z) \left( \int_0^\ell \! \dd z' \, \psi_n^{({\rm D})}(z')( \sqrt{A}f )_1(t',z') + \rho^{-1} \psi_n^{\partial \, ({\rm D})}( \sqrt{A}f)_2(t') \right) \\ 
				\psi_n^{\partial \, ({\rm D})} \left( \int_0^\ell \! \dd z' \, \psi_n^{({\rm D})}(z') (\sqrt{A}f)_1(t',z') + \rho^{-1} \psi_n^{\partial \, ({\rm D})} (\sqrt{A}f)_2(t') \right)
\end{pmatrix}, 
\label{DiriClass.1}
\end{align}
and, moreover,
\beq\label{dddd.ee}
\left\{ \int_0^\ell \! \dd z' \, \psi_n^{({\rm D})}(z')( \sqrt{A}f )_1(t',z') + \rho^{-1} \psi_n^{\partial \, ({\rm D})}( \sqrt{A}f)_2(t') \right\}_{n=1}^\infty \in
 \ell^2(\mathcal N),
\ene
with the norm in $\ell^2(\mathcal N)$ uniformly bounded in $t \in \mathbb R.$ Derivating \eqref{DiriClass.1} with respect to $t$ under the integral and summation signs we obtain, 
 \begin{align}
\frac{\partial}{\partial t}\phi_f^{\rm D}(t,z) = -\int_{\mathbb{R}} \! \dd t' \, \sum_{n = 1}^\infty \frac{1}{\omega_n^{\rm D}} \cos\left(\omega_n^{\rm D}(t-t')\right) 
\begin{pmatrix}
				\psi_n^{({\rm D})}(z) \left( \int_0^\ell \! \dd z' \, \psi_n^{({\rm D})}(z')( \sqrt{A} f )_1(t',z') + \rho^{-1} \psi_n^{({\rm D})}( \sqrt{A} f)_2(t') \right) \\ 
				\psi_n^{\partial({\rm D})} \left( \int_0^\ell \! \dd z' \, \psi_n^{({\rm D})}(z') (\sqrt{A} f)_1(t',z') + \rho^{-1} \psi_n^{\partial \, ({\rm D})} (\sqrt{A}f)_2(t') \right)
				\end{pmatrix}. 
\label{DiriClass.2}
\end{align}
The derivation under the integral an summation signs is justified by the Lebesgue dominated convergence theorem since  each component of  the integrand in \eqref{DiriClass.2} belongs to $\ell^1(\mathcal N)$ with norm in $\ell^1(\mathcal N)$ uniformly bounded in $t',$ and moreover it has compact support in $t'.$  This is true by   \eqref{dddd.ee},  and since  by  \eqref{DiriEigen} and equation (B.11) in \cite{Juarez-Aubry:2020psk}   $|\Psi^{(D)}_n(z)|$  is uniformly bounded in $ z \in [0,\ell],  n \in \mathcal N.$ This proves that $\phi_f^{\rm D}$ is continuously differentiable in $t$ with the derivative  uniformly bounded in  $ t \in \mathbb R$ and $ z \in [0,\ell].$ Similarly, derivating  the first component of \eqref{DiriClass.1} with respect to $z$ under the integral and summation signs we obtain, 
 \begin{align}
\frac{\partial}{\partial z}(\phi_f^{\rm D})_1(t,z) =\!\! -\!\!\int_{\mathbb{R}} \! \dd t' \, \sum_{n = 1}^\infty \frac{1}{(\omega_n^{\rm D})^2} \sin\left(\omega_n^{\rm D}(t-t')\right) \begin{pmatrix}
			\frac{\partial}{\partial z}\psi_n^{({\rm D})}(z) \left( \int_0^\ell \! \dd z' \, \psi_n^{({\rm D})}(z')( \sqrt{A}f )_1(t',z') + \rho^{-1} \psi_n^{\partial \, ({\rm D})}( \sqrt{A}f)_2(t') \right)
\end{pmatrix}. 
\label{DiriClass.3}
\end{align}
We prove that the derivation under the integration and summation signs is justified as in the case of the time derivative, using that by  \eqref{DiriEigen} and equation (B.11) in \cite{Juarez-Aubry:2020psk}  $( \omega_n^{\rm D})^{-1}|\frac{\partial}{\partial z} \psi_n^{({\rm D})}(z)| $  is uniformly bounded  in $ z \in [0,\ell],  n \in \mathcal N.$ This proves that $\phi_f$ is continuously differentiable in $z$ with the derivative  uniformly bounded $ t \in \mathbb R$ and $ z \in [0,\ell].$
\end{proof}

Using prop. \ref{prop:RegDiri} we can safely differentiate $\phi_f^{\rm D}$ and define
\begin{align}
H^{\phi_f^{\rm D}}(t,z) := \frac{1}{2} \left( \begin{pmatrix} 
           (\partial_t (\phi_{f}^{\rm D})_1(t,z))^2   \\
             |\beta_1'| (\partial_t (\phi_{f}^{\rm D})_2(t))^2  \\
         \end{pmatrix} + \begin{pmatrix} 
           (\partial_z (\phi_{f}^{\rm D})_1(t,z))^2  \\
           0  \\
         \end{pmatrix} + \begin{pmatrix} 
           m^2 ((\phi_{f}^{\rm D})_1(t,z))^2   \\
           - (\text{\rm sign} \,\beta'_1 )\,  \beta_2 ((\phi_{f}^{\rm D})_2(t))^2  \\
         \end{pmatrix}  \right) := \begin{pmatrix} 
           H^{(\phi_{f}^{\rm D})_1}(t,z)   \\
           H^{(\phi_{f}^{\rm D})_2}(t)  \\
         \end{pmatrix}.
\label{HClassDiri}
\end{align}

The appearance of the factors $\beta_1$ and $\beta_2$ in the definition of $H^{\phi_{f_2}^{\rm D}}(t,z)$ follows from the interpretation of the coefficients explained below eq. \eqref{Dyn}.

The bulk and boundary local Casimir energies  in the ground state are given respectively by (see \eqref{HCoherent}
)
\begin{subequations}
\label{DiriHCoh}
\begin{align}
 \langle F^{({\rm D})} |  \hat H^{{\rm B}}_{\rm ren}(t,z) F^{({\rm D})} \rangle & = \langle \Omega_\ell^{({\rm D})} |  \hat H^{{\rm B}}_{\rm ren}(t,z) \Omega_\ell^{({\rm D})} \rangle + H^{(\phi_{f}^{\rm D})_1}(t,z), \\
\langle F^{({\rm D})} | \hat H^{\partial}_{\rm ren}(t) F^{({\rm D})} \rangle & = \langle \Omega_\ell^{({\rm D})} | \hat H^{\partial}_{\rm ren}(t) \Omega_\ell^{({\rm D})} \rangle + H^{(\phi_{f}^{\rm D})_2}(t) ,
\end{align}
\end{subequations}
where $\langle \Omega_\ell^{({\rm D})} |  \hat H^{{\rm B}}_{\rm ren}(t,z) \Omega_\ell^{({\rm D})} \rangle$ is given by eq. \eqref{HDiriBulk} and $\langle \Omega_\ell^{({\rm D})} | \hat H^{\partial}_{\rm ren}(t) \Omega_\ell^{({\rm D})} \rangle$ is given by eq. \eqref{DiriHBound} in app. \eqref{App:OldResults}.

At finite, positive temperature $T = 1/\beta$, the bulk and boundary local Casimir energies are given respectively by (see \eqref{HCoherent})

\begin{subequations}
\label{DiriHCohT}
\begin{align}
\langle \hat H^{{\rm B}}_{\rm ren}(t, z) \rangle_{\beta, {\rm D}}^F & = \langle \hat H^{{\rm B}}_{\rm ren}(t, z) \rangle_{\beta, {\rm D}}  + H^{(\phi_{f}^{\rm D})_1}(t,z), \\
\langle  \hat H^{\partial}_{\rm ren}(t) \rangle_{\beta, {\rm D}}^F & = \langle  \hat H^{\partial}_{\rm ren}(t) \rangle_{\beta, {\rm D}} + H^{(\phi_{f}^{\rm D})_2}(t).
\end{align}
\end{subequations}
where $\langle \hat H^{{\rm B}}_{\rm ren}(t, z) \rangle_{\beta, {\rm D}}$ is given by eq. \eqref{HTbulkD} and $\langle  \hat H^{\partial}_{\rm ren}(t) \rangle_{\beta, {\rm D}}$ is given by eq. \eqref{HTboundaryD} in app. \eqref{App:OldResults}.

\subsection{Robin boundary condition at $z = 0$}
\label{subsec:Robin}

For Robin boundary conditions at $z = 0$, the ground state two-point  Wightman function is
\begin{align}
\label{WightmanRobin}
\langle \Omega_\ell^{({\rm R})} | \hat \Phi(t,x) \hat \Phi(t',x') \Omega_\ell^{({\rm R})} \rangle = \sum_{n = 1}^\infty \frac{1}{2\omega_n^{\rm R}} \ee^{-\ii \omega_n^{\rm R}(t-t')} \Psi_n^{({\rm R})} (z) \otimes \Psi_n^{({\rm R})} (z'),
\end{align}
where $\Omega_\ell^{\rm (R)}$ is the ground state in Fock space, and where $(\omega_n^{\rm R})^2 = (s_n^{\rm R})^2 + m^2$ are the eigenvalues of the operator $A$ with $\alpha \neq 0$,  (see \cite[app. B]{Juarez-Aubry:2020psk} for asymptotic estimates of the eigenvalues at large $n$) and with eigenfunctions (normalized with respect to the inner product \eqref{Inner}) given by
\begin{align}
& \Psi_n^{(\rm R)}(z)  = \begin{pmatrix}
				\psi_n^{({\rm R})}(z) \\ 
				\psi_n^{\partial \, ({\rm R})}
\end{pmatrix} = \mathcal{N}_n^{\rm R} 
	\begin{pmatrix}
       \ds   - \frac{\cos \alpha \sin (s_n^{\rm R} z)}{s_n^{\rm R}} +  \sin \alpha \cos( s_n^{\rm R} z)  \\ 
      \ds     \beta_1' \left(\sin \alpha  \cos (\ell s_n^{\rm R})-\frac{\cos \alpha  \sin (\ell s_n^{\rm R})}{s_n^{\rm R}}\right)+\beta_2' (\cos \alpha \cos (\ell s_n^{\rm R})+s_n^{\rm R} \sin \alpha  \sin (\ell s_n^{\rm R}))
    \end{pmatrix},  \nonumber \\
& \mathcal{N}_n^{\rm R}  := \left\{ -\frac{\sin (2 \alpha )}{4 (s_n^{\rm R})^2} - \left[4 (s_n^{\rm R})^2\Big((\beta_1+ \beta_1' m^2)^2+9(\beta_1')^2 (s_n^{\rm R})^4+2 (\beta_1 + \beta_1' m^2) \beta_1' (s_n^{\rm R})^2 \right. \right. \nonumber \\
& \left. \left. +(s_n^{\rm R})^2 \left((\beta_2 + \beta_2' m^2)+\beta_2' (s_n^{\rm R})^2\right)^2\right) \right]^{-1} \left(\left((s_n^{\rm R})^2-1\right) \cos (2 \alpha )-(s_n^{\rm R})^2-1\right) \nonumber \\
    &  \times \left[\ell \left((\beta_1+\beta_1' m^2)^2+\beta_1'^2 (s_n^{\rm R})^4+2 (\beta_1 + \beta_1' m^2) \beta_1' (s_n^{\rm R})^2+(s_n^{\rm R})^2 \left((\beta_2+\beta_2' m^2) +\beta_2' (s_n^{\rm R})^2\right)^2\right) \right.  \nonumber \\
    & \left. \left.  -(\beta_1+\beta_1' m^2) \left((\beta_2+\beta_2' m^2)+3 \beta_2' (s_n^{\rm R})^2\right)+\beta_1' (s_n^{\rm R})^2 \left((\beta_2+\beta_2' m^2) -\beta_2' (s_n^{\rm R})^2\right)\right] \right\}^{-1/2}.
\label{RobinEigen}
\end{align}

At finite temperature $T = 1/\beta > 0$, the two-point Wightman function is given by
\begin{align}\label{we.100R}
\langle\hat \Phi(t,z)  \hat \Phi(t',z') \rangle_{\beta,{\rm R}} &=  \langle \Omega_\ell^{\rm (R)} | \hat \Phi(t,z) \hat \Phi(t',z') \Omega_\ell^{\rm (R)} \rangle + \sum_{n = 1}^\infty \,\frac{1}{2 \omega_n^{\rm (R)}} \,\Psi_n^{\rm (R)} (z)  \otimes \Psi_n^{\rm (R)} (z')\, \,\frac{1}{\ee^{\beta \omega_j^{\rm R}} - 1}\,\left(\ee^{-\ii \omega_n^{\rm R}(t-t')} + \ee^{\ii \omega_n^{\rm R}(t-t')} 
\right).
\end{align}

Both at zero and finite temperature, the bulk renormalized local state polarization and local Casimir energy are computed making use of the bulk-bulk component of \eqref{WightmanRobin} or \eqref{we.100R}. For example, for the two-point  Wightman function \eqref{WightmanRobin}, this is 
\begin{align}
\langle \Omega_\ell^{({\rm R})} | \hat \Phi^{\rm B}(t,z) \hat \Phi^{\rm B}(t',z') \Omega_\ell^{({\rm R})} \rangle & = \sum_{n = 1}^\infty \frac{1}{2\omega_n^{\rm R}} \ee^{-\ii \omega_n^{\rm R}(t-t')} \psi_n^{({\rm R})} (z) \psi_n^{({\rm R})} (z'),
\end{align}
and analogously for the two-point Wightman function at finite temperature \eqref{we.100R}. For the boundary renormalized state polarization and Casimir energy at zero or finite temperature, one makes use the boundary-boundary component of \eqref{WightmanRobin} or \eqref{we.100R}. For example, for the two-point Wigthman function \eqref{WightmanRobin}, this is
\begin{align}
\langle \Omega_\ell^{({\rm R})} | \hat \Phi^\partial(t) \hat \Phi^\partial(t') \Omega_\ell^{({\rm R})} \rangle & = \sum_{n = 1}^\infty \frac{1}{2\omega_n^{\rm R}} \ee^{-\ii \omega_n^{\rm R}(t-t')} \psi_n^{\partial \, ({\rm R})} \psi_n^{\partial \, ({\rm R})},
\end{align}
and analogously for the two-point Wightman function at finite temperature \eqref{we.100R}.

The commutator  in the Robin case  is given by,
\beq\label{causalR}
\mathsf E_{\rm R} \varphi (t,z):= \int_{\mathbb R}\, dt'\, dz'\, E_{\rm R} \left( (t,z), (t',z') \right)\, \varphi(t',z'),
\ene
with the integral kernel,

\begin{align}
E_{\rm R} \left( (t,z), (t',z') \right) = -\sum_{n = 1}^\infty \frac{1}{\omega_n^{\rm R}} \sin\left(\omega_n^{\rm R}(t-t')\right) |\Psi_n^{({\rm R})} (z) \rangle  \langle\Psi_n^{({\rm R})} (z')|.
\label{AdvRetRobin}
\end{align}


\subsubsection{Renormalized local state polarization in coherent and thermal coherent states}

Let $f =  \begin{pmatrix} 
           f_1   \\
           f_2  \\
         \end{pmatrix} \in C_{0,R}(\mathbb{R}, \mathcal H)$, then $F^{(\rm R)} = \ee^{\ii \hat \Phi(f)} \Omega_\ell^{\rm (R)}$ defines a coherent state in the case of a Robin boundary condition at $z = 0$ ($\alpha \neq 0$) around the classical solution $\phi_f^{\rm R} = E_{\rm R} f$, which can be written in detail as 
\begin{align}
\phi_f^{\rm R}(t,z) = -\int_{\mathbb{R}} \! \dd t' \, \sum_{n = 1}^\infty \frac{1}{\omega_n^{\rm R}} \sin\left(\omega_n^{\rm R}(t-t')\right) \begin{pmatrix}
				\psi_n^{({\rm R})}(z) \left( \int_0^\ell \! \dd z' \, \psi_n^{({\rm R})}(z') f_1(t',z') + \rho^{-1} \psi_n^{\partial \, ({\rm R})} f_2(t') \right) \\ 
				\psi_n^{\partial \, ({\rm R})} \left( \int_0^\ell \! \dd z' \, \psi_n^{({\rm R})}(z') f_1(t',z') + \rho^{-1} \psi_n^{\partial \, ({\rm R})} f_2(t') \right)
\end{pmatrix} =: \begin{pmatrix}
(\phi_{f}^{\rm R})_1(t,z) \\
(\phi_{f}^{\rm R})_2(t)
\end{pmatrix}.
\label{RobinClass}
\end{align}

%
%
        
At zero temperature, the bulk and boundary renormalized local state polarizations in the coherent state $F^{\rm (R)}$ are given respectively by (see  \eqref{Coh0} and \eqref{ww.zz})
\begin{subequations}
\label{RobinPolaCoh}
\begin{align}
\langle F^{(\rm R)} | ( \hat \Phi^{\rm B}_{\rm ren} )^2 (t,z) F^{(\rm R)} \rangle & = \langle \Omega^{(\rm R)} | ( \hat \Phi^{\rm B}_{\rm ren} )^2 (t,z) \Omega^{(\rm R)} \rangle + \left((\phi_{f}^{\rm R})_1(t,z)\right)^2, \\
\langle F^{(\rm R)} | ( \hat \Phi^{\partial}_{\rm ren} )^2 (t) F^{(\rm R)} \rangle & = \langle \Omega^{(\rm R)} | ( \hat \Phi^{\partial}_{\rm ren} )^2 (t) \Omega^{(\rm R)} \rangle + \left((\phi_{f})^{\rm R})_2(t)\right)^2,
\end{align}
\end{subequations}
 where we used that for the ground state the one-point function is zero.The bulk and boundary renormalized local state polarizations are given by eq. \eqref{RobinVacuumBulk} and \eqref{RobinVacuumBoundary}, respectively, in app \ref{App:OldResults}.


The bulk and boundary renormalized local state polarizations in the thermal coherent state $\langle \cdot \rangle_{\beta,{\rm D}}^F$ at temperature $T = 1/\beta >0$ are given respectively by (see\eqref{Coh0} and  \eqref{ww.zz})
\begin{subequations}
\label{RobinPolaCohT}
\begin{align}
\langle ( \hat \Phi^{\rm B}_{\rm ren})^2(t,z) \rangle_{\beta,{\rm R}}^F & = \langle ( \hat \Phi^{\rm B}_{\rm ren})^2(t,z) \rangle_{\beta,{\rm R}} + \left((\phi_{f}^{\rm R})_1(t,z)\right)^2, \\
\langle ( \hat \Phi^{\partial}_{\rm ren})^2(t) \rangle_{\beta,{\rm R}}^F & = \langle ( \hat \Phi^{\partial}_{\rm ren})^2(t) \rangle_{\beta,{\rm R}} + \left((\phi_{f})^{\rm R})_2(t)\right)^2,
\end{align}
\end{subequations}
 where we used that for the Gibbs state the one-point function is zero.The  bulk and boundary renormalized local state polarizations at finite temperature are given in app. \ref{App:OldResults}, eq. \eqref{TBulkPolaR} and \eqref{TBoundPolaR},
 respectively.

\subsubsection{Local Casimir energy in coherent and thermal coherent states}

As in the case of the Dirichlet boundary condition, to define the  local Casimir energy we need to differentiate the solution $\phi_f,$ cf. formulae \eqref{HCoherent}, and we need to assume more regularity on $f.$ For this reason, as in the case of Dirichlet boundary  condition we take as  space of test functions
$C_{0,R}(\mathbb R, { D(\sqrt{A})}).$

\begin{prop}\label{prop:RegRobin}

Let $\phi_f^{\rm R} = E_{\rm R} f$ be as above a solution to eq. \eqref{w.5} (with $\alpha \neq 0$)  with
\beq
f =  \begin{pmatrix} 
           f_1   \\
           f_2  \\
         \end{pmatrix} \in C_{0,R}(\mathbb R,  D(\sqrt{A})).
\ene

Then, $\phi_f^{\rm R}(t,z)$ is continuously differentiable in $t$ and $z,$ and the derivatives are uniformly bounded for $ t \in \mathbb R$ and $ z \in [0,\ell].$
\end{prop}
\begin{proof} The proof is analogous to the proof of prop.~\ref{prop:RegDiri}.

\end{proof}

For the local Casimir energy with a Robin boundary condition at $z = 0$ in coherent and thermal coherent states around the classical solution $\phi_f^{\rm R}$ \eqref{RobinClass}, we define in analogy with the Dirichlet case
\begin{align}
H^{\phi_f^{\rm R}}(t,z) := \frac{1}{2} \left( \begin{pmatrix} 
           (\partial_t (\phi_{f}^{\rm R})_1(t,z))^2   \\
           |\beta_1'| (\partial_t (\phi_{f}^{\rm R})_2(t))^2  \\
         \end{pmatrix} + \begin{pmatrix} 
           (\partial_z( \phi_{f}^{\rm R})_1(t,z))^2  \\
           0  \\
         \end{pmatrix} + \begin{pmatrix} 
           m^2 ((\phi_{f}^{\rm R})_1(t,z))^2   \\
           -(\text{\rm sign}\,\beta'_1) \beta_2 ((\phi_{f}^{\rm R})_2(t))^2  \\
         \end{pmatrix}  \right) =: \begin{pmatrix} 
           H^{(\phi_{f})_1^{\rm R}}(t,z)   \\
           H^{(\phi_{f}^{\rm R})_2}(t)  \\
         \end{pmatrix}.
\label{HClassRobin}
\end{align}


The bulk and boundary local Casimir energies at zero temperature are given respectively by (see \eqref{HCoherent})
\begin{subequations}
\label{RobinHCoh}
\begin{align}
 \langle F^{({\rm R})} |  \hat H^{{\rm B}}_{\rm ren}(t,z) F^{({\rm R})} \rangle & = \langle \Omega_\ell^{({\rm R})} |  \hat H^{{\rm B}}_{\rm ren}(t,z) \Omega_\ell^{({\rm R})} \rangle + H^{(\phi_{f}^{\rm R})_1}(t,z), \\
\langle F^{({\rm R})} | \hat H^{\partial}_{\rm ren}(t) F^{({\rm R})} \rangle & = \langle \Omega_\ell^{({\rm R})} | \hat H^{\partial}_{\rm ren}(t) \Omega_\ell^{({\rm R})} \rangle + H^{(\phi_{f}^{\rm R})_2}(t) ,
\end{align}
\end{subequations}
where $\langle \Omega_\ell^{({\rm R})} |  \hat H^{{\rm B}}_{\rm ren}(t,z) \Omega_\ell^{({\rm R})} \rangle$ is given by eq. \eqref{RobinHBulk} and $\langle \Omega_\ell^{({\rm R})} | \hat H^{\partial}_{\rm ren}(t) \Omega_\ell^{({\rm R})} \rangle$ is given by eq. \eqref{RobinHBoundary} in app. \ref{App:OldResults}.

At finite, positive temperature $T = 1/\beta$, the bulk and boundary local Casimir energies are given respectively by (see \eqref{HCoherent})
\begin{subequations}
\label{RobinHCohT}
\begin{align}
\langle \hat H^{{\rm B}}_{\rm ren}(t, z) \rangle_{\beta, {\rm R}}^F  & = \langle \hat H^{{\rm B}}_{\rm ren}(t, z) \rangle_{\beta, {\rm R}}  + H^{(\phi_{f}^{\rm R})_1}(t,z), \\
\langle  \hat H^{\partial}_{\rm ren}(t) \rangle_{\beta, {\rm R}}^F & = \langle  \hat H^{\partial}_{\rm ren}(t) \rangle_{\beta, {\rm R}} + H^{(\phi_{f}^{\rm R})_2}(t),
\end{align}
\end{subequations}
where $\langle \hat H^{{\rm B}}_{\rm ren}(t, z) \rangle_{\beta, {\rm R}}$ is given by eq. \eqref{HTbulkR} and $\langle  \hat H^{\partial}_{\rm ren}(t) \rangle_{\beta, {\rm R}}$ is given by eq. \eqref{HTboundaryR} in app. \ref{App:OldResults}.

\section{The Casimir force}
\label{sec:Force}

Experimentalists are interested in measuring the Casimir force, which is defined as minus the derivative with respect to the interval length, $\ell$, of the expectation value of the integrated (bulk) Casimir energy \cite{Bordag, Fulling:1989nb}. In this section, we explore numerically the Casimir force for the ground state,  with a Dirichlet boundary condition at $z=0,$ which takes the form
\begin{align}
F_{\Omega_\ell^{\rm(D)}}(t,\ell) := - \partial_\ell \left( \int_0^\ell \dd z \langle \Omega_\ell^{\rm(D)} | \hat H_{\rm ren}^{\rm B}(t,z) \Omega_\ell^{\rm(D)} \rangle \right).
\end{align}
Afterwards, we shall comment on the Casimir force  at positive temperature and for coherent and thermal coherent states.

A first purpose of the section is to illustrate how this quantity may be computed numerically for our problem with dynamical boundary conditions. A second purpose will be to see how the Casimir force depends on the values of the coefficients of the problem. As we shall see, the force can be attractive or repulsive depending on these values.

For our first purpose, we consider, for concreteness, the case in which a Dirichlet boundary condition is set at $z = 0$, and set $\beta_1 = -1$, $\beta_1' = 1$, $\beta_2 = 1$, $\beta_2' = -1/2$ and $m^2 = 1$. The values chosen for the coefficients are so that the hyphoteses of Prop. \ref{prop:pos} hold. The ground state local Casimir energy in the bulk is then given by eq. \eqref{HDiriBulk} in app. \ref{App:OldResults}, and is time independent. Thus, the integrated Casimir  energy, given by
\begin{align}
E_{\Omega_\ell^{\rm(D)}}(\ell) := \int_0^\ell \dd z \langle \Omega_\ell | \hat H_{\rm ren}^{\rm B}(z) \Omega_\ell \rangle,
\end{align}
and also the Casimir force is  time independent.

In Fig. \ref{fig:Force} the interval $1 \leq \ell \leq 1.9$ is explored numerically with ten data points spaced evenly with a spacing of width $0.1$, namely we take $\ell \in \{1, 1.1, \ldots, 1.8, 1.9\}$. Each of the data points is a numerical approximation of the integrated Casimir energy by considering the first 50 eigenvalues in each case, but more precise numerical results can be obtained by computing a larger number of eigenvalues. Fig. \ref{fig:Force} provides evidence that in this example the integrated Casimir energy is a non-decreasing function of $\ell$, and as a consequence $F_{\Omega_\ell^{\rm(D)}}(\ell) \leq 0$, i.e., the Casimir force is attractive, as is the case in which Dirichlet boundary conditions are imposed at the two boundaries \cite{Fulling:1989nb}.

\begin{figure}[!ht]
\centering
\includegraphics[width=0.45\textwidth]{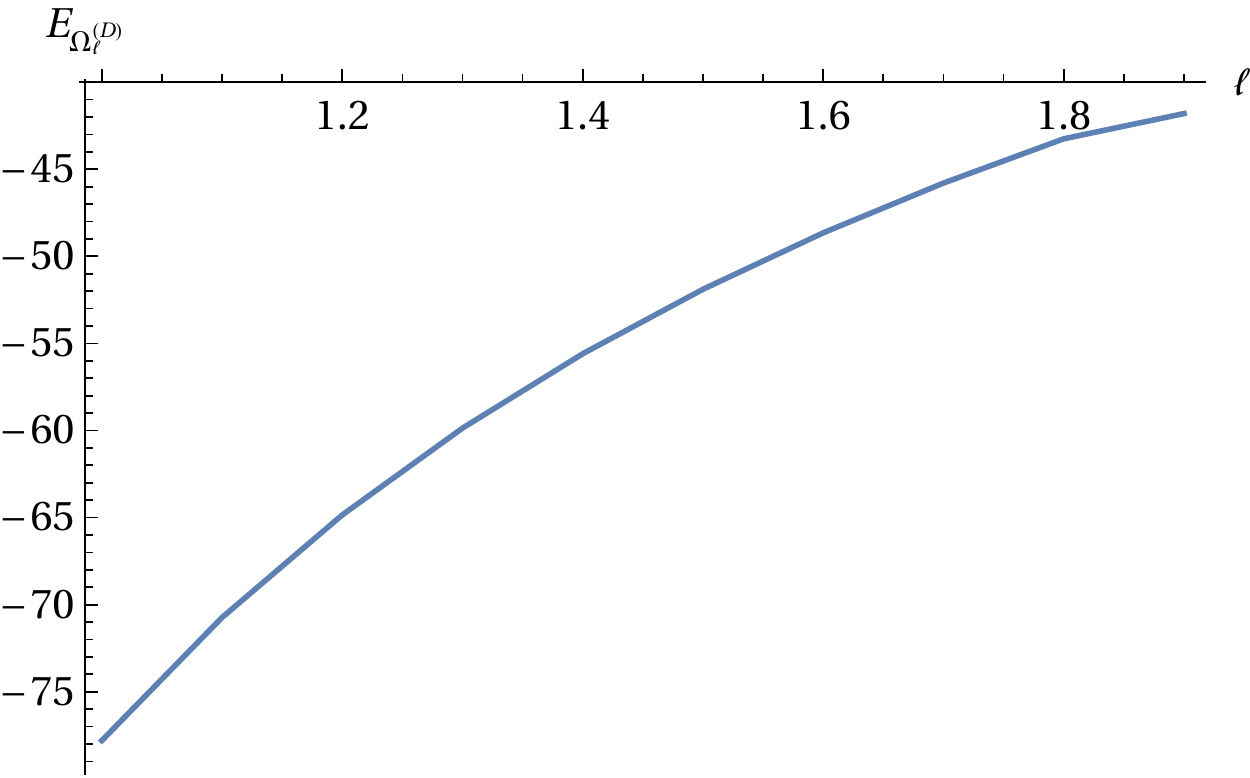}
\caption{Curve fitting of numerical data for the integrated Casimir energy with points at $\ell \in \{1, 1.1, \ldots, 1.9\}$.}
\label{fig:Force}
\end{figure}

Let us now explore the dependence of the Casimir force on the coefficients of the problem. We have already illustrated the massive bulk field case, and we shall now make the (natural) choice $m^2 = 0$ in order to also illustrate the massless situation. It can be seen from eq. \eqref{HDiriBulk} (Dirichlet at $z = 0$) and \eqref{RobinHBulk} (Robin at $z = 0$) that the local Casimir energy is $z$-independent. $z$-independence also occurs  for periodic boundary conditions \cite{Kay:1978zr} and for Dirichlet boundary conditions using the definition of the ``new" renormalized stress-energy tensor \cite{Fulling:1989nb}. Let us consider again, for concreteness, a Dirichlet boundary condition at $z = 0$. We quote for convenience eq. \eqref{HDiriBulk} in the massless case, 
\begin{align}
 \langle \Omega_\ell^{({\rm D})} |  \hat H^{{\rm B}}_{\rm ren} \Omega_\ell^{({\rm D})} \rangle & = \frac{\pi }{48  \ell^2} + \frac{\beta_1' }{2 \pi  \beta_2' \ell}   + \sum_{n = 1}^\infty  \left( \frac{(\mathcal{N}_n^{\rm D})^2 \omega^{\rm D}_n}{4} - \frac{\pi n }{2 \ell^2} + \frac{\pi }{4  \ell^2}  \right).
 \label{BulkCasMassless}
\end{align}

We shall explore numerically the large $\beta_2'$ case. There are two reasons for this. First, that the $\beta_2' \neq 0$ case is a novelty of this paper and of our previous work \cite{Juarez-Aubry:2020psk} with respect to the physics literature quoted in the Introduction, and it is natural to ask oneself what role this parameter plays on the numerical analysis. The second reason is of a more technical nature. The large eigenvalue estimates \cite[eq. (B.11)]{Juarez-Aubry:2020psk} in the massless case,
\begin{align}
\omega_n^{\rm D} = s_n^{\rm D}= \frac{(n-1/2)\pi}{\ell} - \frac{\beta_1'}{\beta_2' \pi (n-1/2)} + O(n^{-3}),
\label{DirichletEstimatesMassless}
\end{align}
become sharper for large $\beta_2'$, and hence numerical analysis becomes more trusworthy in this case. In eq. \eqref{DirichletEstimatesMassless} the coefficient of the $O(n^{-3})$ term is also $O\left((\beta_2')^{-1}\right)$.

We choose the coefficients $ m^2=0,   \beta_1 = 1/200$, $\beta_2 = -1$, $\beta_2' = 100$ and $\beta_1' = -1, -3, -5$. The choice of a small $\beta_1$ is required so that $\rho = \beta_1'\beta_2 - \beta_1 \beta_2' > 0$ whenever $\beta_2'$ is large. The values chosen for the coefficients are so that the hypotheses of Prop. \ref{prop:pos} hold. Fig. \ref{fig:Force2} explores numerically the interval $1 \leq \ell \leq 1.9$ in these three cases. Each curve contains ten data points for the integrated Casimir energy spaced evenly with a spacing of width $0.1$, namely $\ell \in \{1, 1.1, \ldots, 1.8, 1.9\}$. In each case, as before, the integrated Casimir energy is numerically approximated by considering the first 50 eigenvalues.

\begin{figure}[!ht]
\centering
\includegraphics[width=0.45\textwidth]{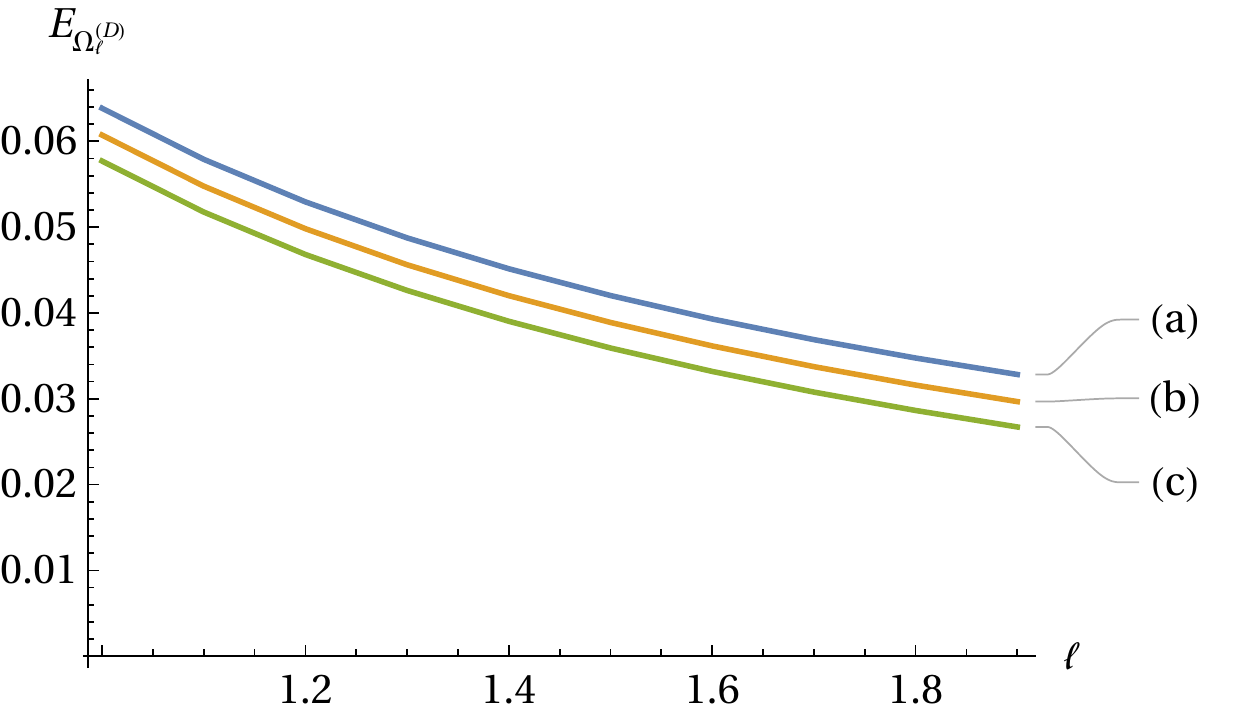}
\caption{Curve fitting of numerical data for the integrated Casimir energy with points at $\ell \in \{1, 1.1, \ldots, 1.9\}$ for (a) $\beta_1' = -1$, (b) $\beta_1' = -3$ and (c) $\beta_1' = -5$.}
\label{fig:Force2}
\end{figure}

It  follows from the numerical exploration presented in Fig. \ref{fig:Force2}  that for the choice $m^2 = 0$, $\beta_1 = 1/200$, $\beta_2 = -1$, $\beta_2' = 100$ and $\beta_1' = -1, -3, -5$ the Casimir force is repulsive. This is also the case notably for anti-periodic boundary conditions \cite{Asorey}.

Let us now consider the case of positive temperature, and as before the Dirichlet boundary condition at $z=0.$ The integrated  Casimir energy is given by,
\beq\label{energy. pos}
E_{\beta, D} := \int_0^\ell \dd z \, \langle \hat H^{{\rm B}}_{\rm ren}(t, z) \rangle_{\beta, {\rm D}},
\ene
where the local bulk Casimir energy $\langle \hat H^{{\rm B}}_{\rm ren}(t, z) \rangle_{\beta, {\rm D}} $ is given by \eqref{HTbulkD}. In Figure \ref{fig:Force3}  we present numerical results for $E_{\beta,D}.$  We choose the coefficients $ m^2=0,   \beta_1 = 1/200$, $\beta_2 = -1$, $\beta_2' = 100$ and $\beta_1' = -1$.  The values chosen for the coefficients are so that the hypotheses of Prop. \ref{prop:pos} hold. Fig. \ref{fig:Force3} explores numerically the interval $1 \leq \ell \leq 1.9$  for temperatures  $T=0, T=1/7, T=1/5,$ and $T=1/3.$  Each curve contains ten data points for the integrated Casimir energy at temperature $T,$ spaced evenly with a spacing of width $0.1$, namely $\ell \in \{1, 1.1, \ldots, 1.8, 1.9\}$. In each case, as before, the integrated Casimir energy at temperature $T$  is numerically approximated by considering the first 50 eigenvalues.  
\begin{figure}[!ht]
\centering
\includegraphics[width=0.45\textwidth]{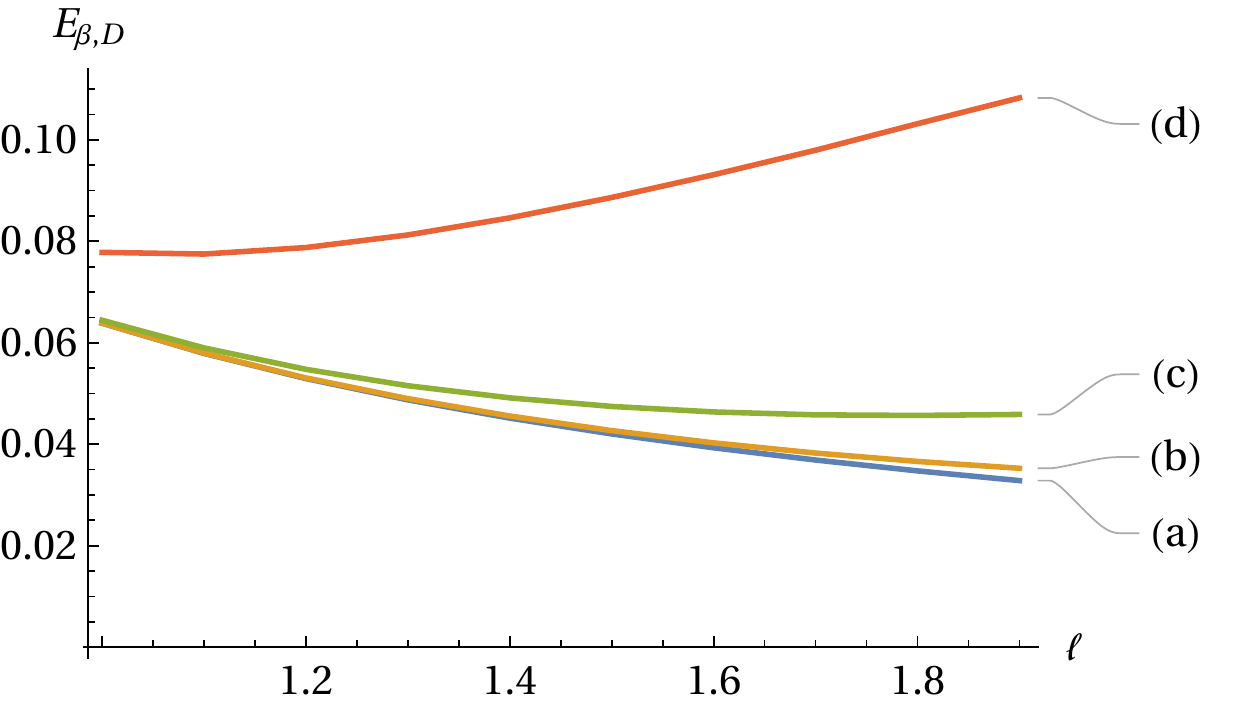}
\caption{Curve fitting of numerical data for the integrated Casimir energy at temperature $T$ with points at $\ell \in \{1, 1.1, \ldots, 1.9\}$ for (a) $T = 0$, (b) $T = 1/7,$  (c) $T= 1/5,$  and (d) T=1/3}
\label{fig:Force3}
\end{figure}
It follows from our numerical exploration in Fig. \ref{fig:Force3} that depending on the temperature $T$  and on the length $\ell,$ the Casimir force, $ - \partial_\ell E_{\beta,D}$ can be repulsive or attractive.

To end the section, we wish to discuss how the Casimir energies and forces are modified when one considers  coherent and thermal coherent states. 
For coherent states, and thermal coherent states, the details of the modifications to the Casimir energy and the Casimir force depend on the classical solution that defines the coherent state. However, it  can be seen from eq. \eqref{DiriHCoh}, \eqref{DiriHCohT}, \eqref{RobinHCoh}, and \eqref{RobinHCohT}   that the difference between the Casimir energy in a coherent state and the ground state, and that the difference between the Casimir energy in a thermal coherent state and a thermal state,   is the classical energy of the solution defining the coherent state, which is non-negative for all values of $\beta_1, \beta_2, \beta_1', \beta_2', m^2$ and $\ell.$

Let us now examine what happens to the Casimir energy and the Casimir force in the particular example of a single mode classical solution and mass zero, $m=0.$ For simplicity we take the Dirichlet boundary condition at  $z=0.$ We take a  function $h \in C_{0, R}(\mathbb{R})$ and consider the classical solution \eqref{DiriClass} with $ f= f_k:= h(t)\, \Psi^{(\rm D)}_k(z),$ for some $k=1,\dots.$   Since $\Psi^{(\rm D)}_k$ is an eigenfunction of $A,$ only the term with $n=k$ is different from zero in the series in \eqref{DiriClass} and we have,
\begin{align}
\phi^{\rm D}_{f_k}(t,z)  = - \int_\mathbb{R} \dd t' h(t') \frac{\sin(\omega_k^{\rm D}(t-t'))}{\omega_k^{\rm D}} \Psi_k^{\rm (D)}(z).
\label{SharpSoln}
\end{align}

Then, by \eqref{DiriHCoh} and \eqref{DiriHCohT}, the  contribution to the integrated Casimir energy given by the classical solution \eqref{SharpSoln} is given by,
\beq\label{appr.10}
E^{\phi^D_{f_k}}:=  \int_0^\ell \,\dd z  H^{(\phi^{\rm D}_{f_k})_1}.
\ene
{\color{black} In Figure \ref{fig:Force4}  we present numerical results for $ E^{\phi^D_{f_k}},$ for the lowest eigenvalue ($k=1$), $\omega_k = \omega_1 \approx 1.5771$, and for the function $h(t)= \cos t ,$ for $|t|\leq \pi/2,$ 
 and $h(t)=0,$ for $|t|\geq \pi/2.$ We choose the coefficients $ m^2=0,   \beta_1 = 1/200$, $\beta_2 = -1$, $\beta_2' = 100$ and $\beta_1' = -1$.  The values chosen for the coefficients are so that the hypotheses of Prop. \ref{prop:pos} hold. 
Each curve in Fig. \ref{fig:Force4} represents the (time-dependent) total energy of the classical solution \eqref{SharpSoln} at 
 times, $t = 0$,  $t = \pi/2,$ $t=\pi,$  $t= 3\pi/2$  and  $t= 2\pi$ as the interval length $\ell$ varies.}   
 \begin{figure}[!ht]
\centering
\includegraphics[width=0.45\textwidth]{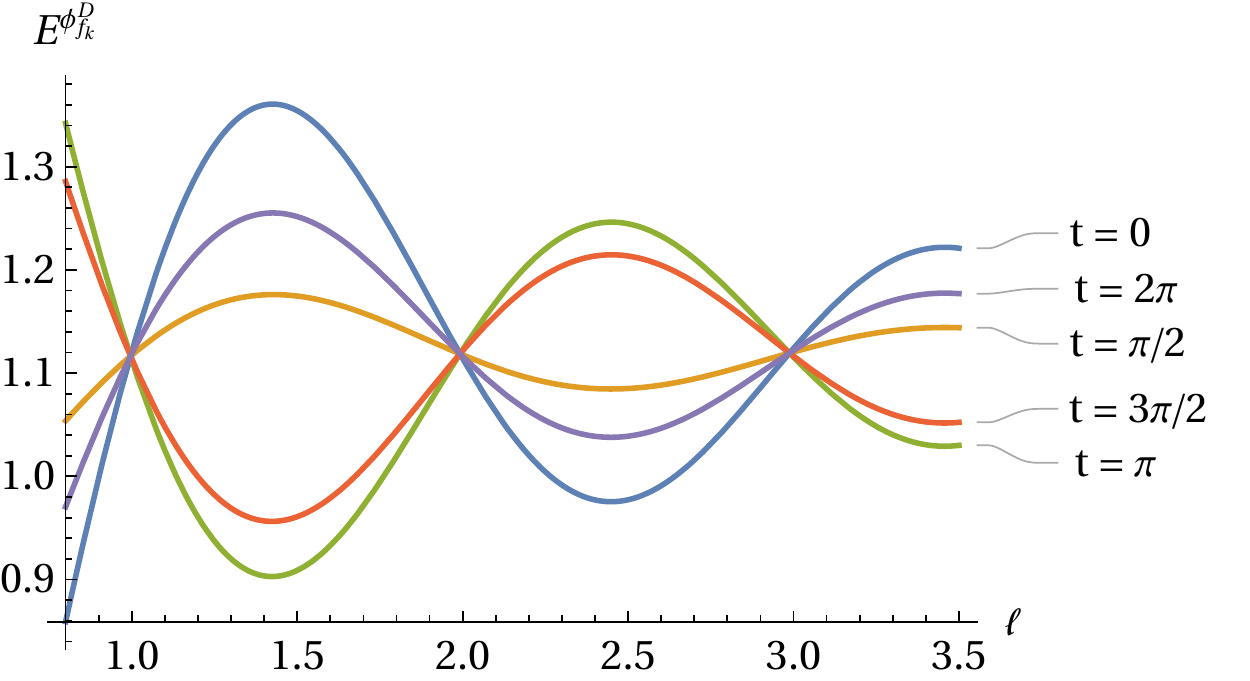}
\caption{Curve fitting of numerical data for the   contribution to the integrated Casimir energy given by the classical solution \eqref{SharpSoln} with points at $\ell \in \{1, 1.1, \ldots, 1.9\}$  for the lowest eigenvalue, $k=1,$ and for times, $t = 0$,  $t = \pi/2,$ $t=\pi,$  $t= 3\pi/2,$  and  $t= 2\pi.$}
\label{fig:Force4}
\end{figure}

It follows from our numerical exploration in Fig. \ref{fig:Force4} that depending on the time $t$  and on the length $\ell,$ the
contribution to the  Casimir force, $-\partial_\ell E^{\phi^D_{f_k}},$ given by the classical solution \eqref{SharpSoln}  can be repulsive or attractive.

\section{The causal propagator and the advanced and retarded Green's operators}
\label{sec:F-loc}
In this section we consider the  causal propagator and the advanced and retarded Green's operators. As we will see,  the integral operator, $\mathsf E,$  defined in \eqref{daaa.bbb} with the integral kernel $E(t-t')$  given  in \eqref{aaa.bbb}, is the causal  propagator for  our theory, when restricted to a convenient domain. Recall that  by \eqref{dd.zz.ww},  the integral operator  $\mathsf E$   gives the commutator formula for  the quantum fields.

We first introduce some appropriate definitions.  For every $\varphi \in D(A)$ we define the graph norm of $\varphi$ as follows,
\beq
\|\varphi\|_{D(A)}:= \|\varphi\|_\mathcal H + \| A \varphi\|_\mathcal H.
\ene

The domain of $A$ endowed with the graph norm is a Hilbert space. Below we always assume that $D(A)$ is endowed with the graph norm. For $j=1,\dots,$ we designate by 
$C^j(\mathbb R, D(A))$ the functions from $\mathbb R$ into $D(A)$ that have $j$ continuous derivatives, and  by  $C^j_0(\mathbb R, D(A))$ the functions in  $C^j(\mathbb R, D(A))$ that have compact support in $\mathbb R.$  Further, by 
$C(\mathbb R, D(A))$ we denote the continuos functions from $\mathbb R$ into $D(A),$ and  by  $C_0(\mathbb R, D(A))$ the functions in  $C(\mathbb R, D(A))$ that have compact support in $\mathbb R.$ 

For any set $O \subset  \mathbb R \times  [0,\ell]$ we denote by $J_\pm(O),$ respectively, the causal future and the causal past of $O,$ and by $J(O):= J_+(O) \cup J_-(O)$ the union of the causal past and future of $O.$  We designate by $ \mathcal B(\mathcal H)$ the Banach space of all bounded operators in $\mathcal H.$


 Further, we denote by $Q_A$ the Klein-Gordon operator,
\beq
Q_A:=\frac{\partial^2}{\partial_t^2}+ A.
\ene
\begin{defn}
We say that an operator, $\mathsf E,$ from $C^2_0(\mathbb R, D(A))$ into $C^2(\mathbb R, D(A)),$ is an integral operator, and that $E(t-t')$  is the integral kernel of  $ \mathsf E,$ if 
\beq
(\mathsf Ef)(t,z)= \int_{\mathbb R}\, dt'\,( E(t-t')\, f(t'))(z),
\ene
where $E(t), t \in \mathbb R,$ is a bounded operator on $\mathcal H.$ Furthermore,  we assume that  $E \in C^ 2(\mathbb R, \mathcal B(\mathcal H)).$       
\end{defn}
We now give our definition of causal propagator. 
    
\begin{defn}\label{causalprop}
An integral operator, $\mathsf E,$ from $C^2_0(\mathbb R, D(A))$  into  $C^2(\mathbb R, D(A))$ is a causal propagator if the following conditions are satisfied.
\begin{enumerate}[(i)]
\item 
\beq 
Q_A \mathsf E f =\mathsf  E Q_A f=0 , \qquad f \in C^2_0(\mathbb R, D(A)),
\ene
\item 
\beq
\text{\rm supp} \, \mathsf  (E f)_1 \subset  J ( {\rm supp}\,  f_1),  \qquad f \in C^2_0(\mathbb R, D(A)),
\ene
where for any function $g$ we denote by ${\rm supp}\, g$ the support of $g.$

\item The integral kernel of $\mathsf E$ satisfies,
\beq
E(0)=0, \qquad \partial_t E(0)= -I.
\ene
\end{enumerate}
\end{defn}
\begin{thm} \label{causal} Assume that   the assumptions of Proposition ~\ref{prop:pos} hold and that $ \beta_2 \neq 0.$ Then, the integral operator \eqref{daaa.bbb} with integral kernel \eqref{aaa.bbb},  from $C^2_0(\mathbb R, D(A))$  into  $C^2(\mathbb R, D(A))$ is a causal propagator.
\end{thm}
\begin{proof}
Condition (i) of Definition~ \ref{causalprop} is immediate from the functional calculus of self-adjoint operators, and assumption (iii) follows from the definition of $E(t).$

Note that,
\beq
(E f)(t,z) = \int_{\mathbb R}\, dt' \phi_{t'}(t,z), 
\ene
where for each fixed $ t' \in \mathbb R,$
\beq
\phi_{t'}(t,z):= (E(t-t')f(t'))(t,z).
\ene 
Observe that $\phi_{t'}(t,z)$ satisfies,
\beq\label{wxx.1}
Q_A \phi_{t'}(t,z)=0,
\ene
and that
\beq
\phi_{t'}(t',z)=0, \qquad   \partial_t \phi_{t'}(t',z)= - \phi(t',z). 
\ene

We prepare the following results that we use later in the proof.
\begin{enumerate}[(a)]

\item
Consider the causal diamond, $D(I),$ in  $\mathbb{R} \times [0, \ell]$ of width $I = [a,b] \subset [0, \ell]$ at $t = 0$, given by,
\beq
D(I):= \left\{ (t, z) \in [  \frac{a-b}{2},  \frac{b-a}{2}]\times [0,\ell]: z \in [a+|t|, b-|t|]\right\}.
\ene
 We show that if $(\phi_{0})_1(0,z)$  is identically zero on $I$, then $(\phi_{0})_1(t,z)$ vanishes on the whole of $D(I).$

To each constant-$t$ cross section of the diamond, $I(t):=[a+|t|, b-|t|] $, we  associate a positive energy functional
\begin{align}
\epsilon^{I(t)}_{\rm B}[(\phi_0)_1](t) & :=  \frac{1}{2} \Big(|| (\dot \phi_0)_1(t,z)||^2_{L^2(I(t), \dd z)} + ||( \partial_z \phi_0)_1(t,z) ||^2_{L^2(I(t), \dd z)} + m^2 ||( \phi_0)_1(t,z) ||^2_{L^2(I(t), \dd z)} \Big). \label{EnergyFunctionalBulk}
\end{align}

Consider the situation in which $I(t)$ is to the future of $I$, whence it takes the form $I(t) = [a+t, b-t]$ with $t> 0$. We compute the time derivative of $\epsilon^{I(t)}_{\rm B}[\phi_0](t)$ and show it is non-positive. To simplify the notation we denote, $u_1(t,z):=(\phi_0)_1(t,z).$ Then, we have

\begin{align}
\frac{\dd}{\dd t} \epsilon^{I(t)}_{\rm B}[u_1](t) & = \int_{a+t}^{b-t} \dd z \left( \dot u_1 \ddot u_1 + \partial_z \dot u_1 \partial_z u_1 + m^2 \dot u_1 u_1\right) - \frac{1}{2} \left[ \dot u_1^2 + \partial_z u_1^2 + m^2 u_1^2 \right]_{z = b-t} - \frac{1}{2} \left[ \dot u_1^2 + \partial_z u_1^2 + m^2 u_1^2 \right]_{z = a+t} \nonumber \\
& = - \frac{1}{2} \left[ \dot u_1^2 + \partial_z u_1^2 + m^2 u_1^2 - 2 \dot u_1 \partial_z u_1\right]_{z = b-t} - \frac{1}{2} \left[ \dot u_1^2 + \partial_z u_1^2 + m^2 u_1^2 + 2 \dot u_1 \partial_z u_1 \right]_{z = a+t} \nonumber \\
& = - \frac{1}{2} \left[ \left(\dot u_1 - \partial_z u_1\right)^2 + m^2 u_1^2 \right]_{z = b-t} - \frac{1}{2} \left[ \left(\dot u_1 + \partial_z u_1\right)^2 + m^2 u_1^2 \right]_{z = a+t} \leq 0.
\label{BEF2}
\end{align}
Hence, $\epsilon^{I(t)}_{\rm B}[u_1](t) $ is non increasing for $   t\geq 0,$ and as $\epsilon^{I(0)}_{\rm B}[u_1](0)=0, $ it follows that $\epsilon^{I(t)}_{\rm B}[u_1](t)=0,$ for $ t \geq 0.$  In a similar way   we prove that $\frac{\dd}{\dd t} \epsilon^{I(t)}_{\rm B}[u_1](t) \geq 0,$ for $ t \leq 0,$ and as before we prove that   $\epsilon^{I(t)}_{\rm B}[u_1](t)=0,$ for $ t \leq 0.$  Hence,  $\epsilon^{I(t)}_{\rm B}[u_1](t)=0,$ in the whole diamond $D(I)$  and then, $u_1(t,z)$ is identically zero on $D(I).$

\item
Suppose now that the interval $I$ meets the right-end boundary, thus $I = [a, \ell]$. In this case, one is not concerned with the diamond region $D(I)$, but with the half 
diamond 
\beq
D_{\rm R}(I):= \left\{ (t, z) \in [  a-\ell,  \ell-a]\times [0,\ell]: z \in [a+|t|, \ell ]\right\}.
\ene
 Associated to any constant-$t$ cross-section, $I(t),$  of $D_{\rm R}(I),$
 \beq
 I(t):= \{ (t,z) \in D_{\rm R}(I) :  z \in [a+|t|, \ell ] \},
 \ene
we  define the energy functional
\begin{subequations} 
\begin{align}
\epsilon^{I(t)}_{\rm R}[u_1](t) & := \epsilon^{I(t)}_{\rm B}[u_1](t) + \epsilon_{\ell}[u_1](t), \quad \text{with} \label{EnergyFunctional} \\
 \epsilon_{\ell}[u_1](t) & := \frac{1}{2} \left[ - \frac{\beta_1}{\beta_2} (u_1(t,\ell))^2 + \frac{1}{\rho}\left( \frac{\beta_1 \beta_2'}{(-\beta_2 \beta_2')^{1/2}} u_1(t,\ell) + (-\beta_2 \beta_2')^{1/2}\partial_z u_1(t, \ell) \right)^2  \right. \nonumber \\ & + \left.  \frac{1}{\rho} \left( -\beta_1' \dot u_1(t,\ell) + \beta_2' \partial_z \dot u_1(t, \ell) \right)^2 \right], \label{EnergyFunctionalBoundary}
\end{align}
\end{subequations}
and $\epsilon^{I(t)}_{\rm B}[u_1](t)$ defined by \eqref{EnergyFunctionalBulk}. Under our  assumptions  the boundary energy functional $\epsilon_{\ell}[u_1](t)$ \eqref{EnergyFunctionalBoundary} is non-negative.

We begin by considering the bulk contribution \eqref{EnergyFunctionalBulk} to the energy functional \eqref{EnergyFunctional}. On any constant-$t$ cross section of $D_{\rm R}(I)$ to the future of $I = [a,\ell]$, that  we denote $I(t) := [a+t,\ell]$ with $0 \leq t \leq \ell-a$, we have by an adaptation of the argument  used to prove  \eqref{BEF2}
(where the upper integration limit is now $z = \ell$ in  \eqref{EnergyFunctionalBulk}) that
\begin{equation}
\frac{\dd}{\dd t} \epsilon^{I(t)}_{\rm B}[u_1](t) \leq \dot u_1(t, \ell) \partial_z u_1(t, \ell).
\label{EnergyFuncBound1}
\end{equation}

For the boundary contribution of the energy functional \eqref{EnergyFunctionalBoundary} we have that
\begin{align}
\frac{\dd}{\dd t} \epsilon_{\ell}[u_1](t) & =  - \frac{\beta_1}{\beta_2} u_1(t,\ell) \dot u_1(t,\ell) + \frac{1}{\rho} \left( -\beta_1' \dot u_1(t,\ell) + \beta_2' \partial_z \dot u_1(t, \ell) \right)\left( -\beta_1' \ddot u_1(t,\ell) + \beta_2' \partial_z \ddot u_1(t, \ell) \right) \nonumber \\
& + \frac{1}{\rho}\left( \frac{\beta_1 \beta_2'}{(-\beta_2 \beta_2')^{1/2}} u_1(t,\ell) + (-\beta_2 \beta_2')^{1/2}\partial_z u_1(t, \ell) \right) \left( \frac{\beta_1 \beta_2'}{(-\beta_2 \beta_2')^{1/2}} \dot u_1(t,\ell) + (-\beta_2 \beta_2')^{1/2}\partial_z \dot u_1(t, \ell) \right).
\end{align}

Using the Klein-Gordon equation \eqref{wxx.1}  and the boundary condition at $z = \ell$ in \eqref{DomA}, we have that
\begin{align}
\frac{\dd}{\dd t} \epsilon_{\ell}[u_1](t) & = -  \dot u_1(t, \ell) \partial_z u_1(t, \ell).
\label{EnergyFuncBound2}
\end{align}

Thus, we have from \eqref{EnergyFunctional}, \eqref{EnergyFuncBound1} and \eqref{EnergyFuncBound2} that throughout the future component of half-diamond region
\begin{equation}
\frac{\dd}{\dd t} \epsilon^{I(t)}_{\rm R}[u_1](t) \leq 0.
\end{equation}

Thus, vanishing initial data yields a vanishing solution throughout the component of $D_{\rm R}(I)$ to the future of $I$. The converse argument holds for the past component of $D_{\rm R}(I)$.

Hence, if at $t = 0$, we have vanishing data $u(0,z) = \dot u(0,z)$ for $z \in [a, \ell]$, the solution must vanish in whole half-diamond region $D_{\rm R}(I).$ 
\item
Suppose now that one is interested in the left half diamond of width $I = [0, a]$ intersecting the left boundary at $z = 0$. This forms the right half-diamond,
\beq
D_{\rm L}(I):= \left\{ (t, z) \in [  -a,  a]\times [0,\ell]: z \in [0, a-|t| ]\right\}.
\ene

An analogous argument using the energy functional associated to constant-$t$ cross sections of $D_{\rm L}(I),$
\beq
 I(t):= \{ (t,z) \in D_{\rm L}(I) :  z \in [0, a-|t|] \},
\ene
\begin{subequations} 
\begin{align}
\epsilon^{I(t)}_{\rm L}[u_1](t) & := \epsilon^{I(t)}_{\rm B}[u_1](t) + \epsilon_{0}[u_1](t), \quad \text{with} \label{EnergyFunctionalL}\\
\epsilon_{0}[u_1](t) & := -\frac{1}{2} \cot\alpha  \,  (u_1(t,0))^2,   \frac{\pi}{2}\leq    \alpha < \pi,  \qquad  \epsilon_{0}[u_1](t)=0, \alpha=0,
\label{EnergyFunctionalBoundaryRobin}
\end{align}
\end{subequations}
which under our assumptions is positive, implies that for vanishing initial data on $I$, the solution vanishes throughout the region $D_{\rm L}(I)$.

%
%
%

\end{enumerate}

Let us go back to the proof of  (ii). It is enough to
 prove that if $(t,z) \notin  J ({\rm supp}\, f_1),$ then, $ (\mathsf E f)_1(t,z) =0.$ To fix the ideas let us consider $(t,z)= (0, \ell/2).$ The other cases are proven in a similar  way. We have,
\begin{align}
J(\{(0, \ell/2)\})= [-\ell/2, -\infty]\times [0,\ell] \cup \{ (t',z)\in [-\ell/2, 0] \times [0,\ell]: z \in [ \ell/2- |t'|, \ell /2+|t'|] \}\cup \\
 \{ (t',z)\in [ 0, \ell/2] \times [0,\ell]: z \in  [\ell/2-t', \ell/2+t']\}
\cup  [\ell/2, \infty)\times  [0,\ell].   
 \end{align}
Then, if $(0, \ell/2)\notin J ( {\rm supp}\, f_1),$ we must have that $ f_1(t',z)=0,$ for $(t',z) \in  J(\{(0, \ell/2)\}).$ However if $f_1(t',z)=0$, for $(t',z) \in  [-\ell/2, -\infty]\times [0,\ell]  \cup  [\ell/2, \infty)\times  [0,\ell],$ $\phi_{t'}(t,z)=0$ by the uniqueness of the Cauchy problem for the Klein-Gordon equation. Further, if    $f_1(t',z)=0$, for $(t',z) \in  \{ (t',z)\in [-\ell/2, 0] \times [0,\ell]: z \in [ \ell/2- |t'|, \ell /2+|t'|] \}\cup \\
\{(t',z)\in [0,\ell
 \{ (t',z)\in [ 0, \ell/2)] \times [0,\ell]: z \in  [\ell/2-t', \ell/2+t'] \}$ then $\phi_{t'}(t,z)=0$ by  item (a) above. In the case of a general point $(0,z), z \in [0,\ell],$ we prove
 the 
same result using also items (b) and (c), and for points $(t,z)$ with $t \neq 0$ we use a translation in time.   Then, if  $(t,z) \notin  J ( {\rm supp}\, f_1),$ we have that, $ \mathsf( E f)_1(t,z) =0.$ 
\end{proof}

\begin{rem}{\rm
 For  $f, g \in C^2_0(\mathbb R, D(A)),$  with $f_1$ and $g_1$  causally disjoint, 
 i.e., that  are spacelike separated,  the commutator of $\hat{\Phi}(f)$ and $\hat{\Phi}(g)$ is zero. For this purpose, by \eqref{dd.zz.ww}, it is enough to prove that,
\beq\label{spd}
\int_{\mathbb R}\, dt\,   ((\mathsf E f)(t,\cdot),g(t,\cdot))_{\mathcal H}=0.
\ene
By Theorem ~\ref{causal} and since $f_1$ and $g_1$ are spacelike separated,
\beq\label{spd2}
\int_{\mathbb R}\, dt\, \int_{[0,\ell] } \, dz\, (\mathsf E f)_1(t,z)\,g_1(t,z)=0.
\ene
Furthermore, since $ f(t,\cdot), g(t,\cdot) \in D(A),$ by the boundary condition at $z=\ell$ in \eqref{DomA}
\beq\label{spd.1}
(Ef)_2(t)\, g_2(t)= (\beta_1' ( Ef)_1(t,\ell)- \beta_2' (\partial_z  (Ef )_1(t,\ell))\, (\beta_1' g_1(t,\ell) -\beta_2' \partial_z g_1(t,\ell)).
\ene
Since $f_1$ and $g_1$ are spacelike separated, the right-hand side of \eqref{spd.1} is zero, and then,  using also \eqref{spd2} we prove that \eqref{spd} holds.
}\end{rem}
We now introduce our definition of advanced  and retarded Green's operators.
\begin{defn}\label{greenop}
The  integral operator, $\mathsf E_{\rm ret},$ respectively, $\mathsf E_{\rm adv},$ from $C^2_0(\mathbb R, D(A))$  into  $C^2(\mathbb R, D(A))$ is a retarded Green's operator, respectively advanced Green's operator, if the following conditions are satisfied.
\begin{enumerate}[(i)]
\item 
\begin{align} 
Q_A \mathsf E_{\rm ret} f =\mathsf  E_{\rm ret} Q_A f=f , \qquad f \in C^2_0(\mathbb R, D(A)),\\
Q_A \mathsf E_{\rm adv} f =\mathsf  E_{\rm adv} Q_A f=f , \qquad f \in C^2_0(\mathbb R, D(A)).
\end{align}
\item 
\beq
\text{\rm supp} \, \mathsf  (E_{\rm ret} f)_1 \subset  J_- ( {\rm supp}\,  f_1),  \text{\rm supp} \, \mathsf  (E_{\rm adv} f)_1 \subset  J_+ ( {\rm supp}\,  f_1),  \qquad f \in C^2_0(\mathbb R, D(A)).
\ene
\end{enumerate}
\end{defn}
We define the following   integral operators, $\mathsf E_{\rm ret}, \mathsf E_{\rm adv}, $ from $C^2_0(\mathbb R, D(A))$  into  $C^2(\mathbb R, D(A)),$
\begin{subequations}
\begin{align}\label{argop}
(\mathsf E_{\rm ret}f)(t,z)& = \int_t^\infty\,dt'\,( E(t-t')\, f(t'))(z),\\
 (\mathsf  E_{\rm adv}f)(t,\cdot)& =  -\int_{-\infty}^t\,dt'\, (E(t-t')\, f(t'))(z), \qquad f \in C^2_0(\mathbb R, D(A)).\label{argop2}
\end{align}
\end{subequations}
where $E(t), t \in \mathbb R,$ is defined in \eqref{aaa.bbb}.

\begin{thm} Under the conditions of Theorem~\ref{causal} the  integral operator $\mathsf E_{\rm ret},$ defined in \eqref{argop}, is a  retarded  Green's operator, and the integral operator $\mathsf E_{\rm adv},$ defined in \eqref{argop2}, is an   advanced Green's operator.
\end{thm}
\begin{proof}
condition (i) of Definition~\ref{greenop} is immediate by the functional calculus of self-adjoint operators.  Moreover, condition (ii) of Definition~\ref{greenop}  follows from Theorem~\ref{causal} and condition (ii) of Definition~\ref{causalprop}.
\end{proof} 
Finally we prove the uniqueness of the advanced and retarded Green's operators and of the causal propagator.
\begin{thm}
Under the assumptions of Theorem~\ref{causal} the advanced and retarded Green's operators and the causal propagator are unique.
\end{thm} 
\begin{proof}

Let us first consider the  retarded Green's operator. Suppose that $\mathsf E_{\rm ret,1}$ and $\mathsf E_{\rm ret, 2}$ are retarded Green's operators. Denote,
 \beq
 \Delta \mathsf E_{\rm ret}:= \mathsf E_{\rm ret,1}- \mathsf E_{\rm ret ,2}.
 \ene
 Then, we have,
 \beq
 Q_A\, \Delta \mathsf E_{\rm ret}=0.
 \ene
 Moreover, as $\mathsf E_{{\rm ret} ,i}, i=1,2$ are retarded for any $ f \in C^2_0(\mathbb R, D(A)),$ for some $t_0 \in \mathbb R,$
 \beq
 (\Delta \mathsf E_{\rm ret} ) f(t,\cdot) =0,\qquad  t \geq t_0.
 \ene
 Furthermore,  for any $ g \in C^2_0(\mathbb R, D(A)),$ and with $\mathsf E_{\rm adv}$ the advanced Green's operator given,
  in \eqref{argop2} 
 \beq
 \mathsf E_{\rm adv }g(t,\cdot)=0, \qquad  t \leq t_1,
  \ene
  for  some $t_1 \in \mathbb R.$ It follows that the scalar product in $\mathcal H,$ $ (\Delta \mathsf E_{\rm ret} f(t,\cdot), \mathsf E_{\rm adv}\,g(t,\cdot))_{\mathcal H}$ has compact support in $t \in \mathbb R.$ Hence, we can integrate by parts in $t$
  in the identity below,
  \begin{align}
  0= \int_{\mathbb R}\, dt\,   (Q_A \Delta \mathsf E_{\rm ret } f(t,\cdot), \mathsf E_{\rm adv} g(t,\cdot))_{\mathcal H}=   \int_{\mathbb R}\, dt\,   ( \Delta \mathsf E_{\rm ret} f(t,\cdot),Q_A \mathsf E_{\rm adv} g(t,\cdot))_{\mathcal H}=   \int_{\mathbb R}\, dt\,   ( \Delta \mathsf E_{\rm ret } f(t,\cdot), g(t,\cdot))_{\mathcal H}.  
  \end{align} 
  Hence, we conclude that,
  \beq
  \int_{\mathbb R}\, dt\,   ( \Delta \mathsf E_{\rm ret} f(t,\cdot), g(t,\cdot))_{\mathcal H}=0, \qquad f,g \in C^2_0(\mathbb R, D(A)),  
\ene
and then, $\Delta \mathsf E_{\rm ret} =0.$ We prove in a similar y that the advanced Green's operator is unique. Moreover let $\mathsf E$ be a causal propagator with integral kernel 
$E(t-t').$ Denote by $\mathsf E_{\rm ret }$ the retarded  Green's operator with integral kernel  $\Theta(t'-t)\, E(t-t')$
 and by $\mathsf E_{\rm adv}$ the advanced Green's operator with integral kernel
$ - \Theta(t-t')\,E(t-t').$  
Here $ \Theta$ is the  Heaviside function. Hence,
\beq
\mathsf E= \mathsf E_{\rm ret} - \mathsf E_{\rm adv},
\ene
and as $\mathsf E_{\rm ret}, \mathsf E_{\rm adv}$ are unique, it follows that $\mathsf E$ is unique.
\end{proof}
The causal propagator and the advanced and retarded Green's operators were considered in \cite{Dappiaggi:2018jsq}. Note, however, that for the construction of the causal propagator and the advanced and retarded Green's operators it is crucial that these operators are defined both for functions with support in the interior of the spacetime,  and for functions that  have support on the boundary of the spacetime and that satisfy the boundary condition, i.e., that are in $D(A).$ This is also crucial for the uniqueness of the causal propagator and of the advanced and retarded Green's operators. Note that the causal propagator and the advanced and retarded Green's operators for different boundary conditions are different, and  that they are causal propagator and advanced and retarded Green's operators when restricted to the interior of the spacetime.

As is well known, the causal propagator and the advanced and retarded Green's operators are unique  in globally hyperbolic spacetimes \cite{Lichne}.  Note that $C_0^\infty((0, \ell)))\oplus\{0\} \subset D(A)$ and that our causal propagator remains causal when restricted to a globally hyperbolic subspacetime, and, hence, it has to coincide with the  causal propagator of the globally hyperbolic subspacetime. This means that our theory is F-local in the sense of  Kay \cite{Kay:1992es} (see also \cite{Fewster:1995bu}).

\section{Final remarks}
\label{sec:Conc}

We have extended our previous work \cite{Juarez-Aubry:2020psk}, where the bulk and boundary renormalized local state polarizations and local Casimir energies in the ground state and at finite positive temperature were obtained for a mixed bulk-boundary system, consisting of a scalar field defined on the interval $[0, \ell]$ coupled to a boundary observable defined on the right end of the interval. In this work, we have obtained these local renormalized state polarizations and Casimir energies in coherent and thermal coherent states.

We have found that the difference between the local Casimir energy in the ground state and in a coherent state is given by the classical energy of the configuration around which the coherent state is ``peaked", and similarly at finite temperature. This situation is in analogy to what occurs for coherent states in globally hyperbolic spacetimes, where the stress-energy tensor in a coherent state consist of a classical component plus a quantum component. We have also explored numerically the Casimir force.

We have also studied some of the structural properties of the theory defined by our bulk-boundary system.  We have proved that our theory has a unique causal propagator and unique advanced and retarded Green's operators. The causal propagator gives the commutation relations of the quantum fields, and it is F-local.

As we have commented in \cite{Juarez-Aubry:2020psk}, several generalizations of this work can be pursued. First, we have studied the bulk-boundary mixed system on a interval, but in principle one could generalise this to $n$ dimensions with the aid of the Fourier transform and controlling certain integrals. Second, other linear fields can be studied by similar methods. Third, we have studied the static Casimir effect, but the dynamical Casimir effect is also of great relevance. It is especially interesting the case in which the coefficients $\beta_1$, $\beta_1'$, $\beta_2$ and $\beta_2'$ are time-dependent, in which case we should expect particle creation.
\section*{Acknowledgments}
Benito A. Ju\'arez-Aubry is supported by a DGAPA-UNAM Postdoctoral Fellowship. 
We thank Claudio Dappiaggi for some useful discussions. 
\appendix

\section{Renormalized local state polarization and local Casimir energy in the ground state and at finite temperature}
\label{App:OldResults}

In this appendix, we collect results  in the quantization of our mixed bulk-boundary system,  that we obtained in \cite{Juarez-Aubry:2020psk}, and that we used in sec. \ref{sec:DynBC}. More precisely, we present the bulk and boundary renormalized local state polarizations and local Casimir energies for the cases of a Dirichlet ($\alpha = 0$) and  Robin ($\alpha \neq 0$) boundary condition at $z = 0$.

\vskip-5cm
\subsection{Renormalized local state polarization for a Dirichlet boundary condition at $z = 0$}

\subsubsection{At zero temperature}

The bulk renormalized local state polarization in the ground state is
\begin{align}
\langle \Omega_\ell^{({\rm D})} |  ( \hat \Phi^{{\rm B}}_{\rm ren})^2(t,z) \Omega_\ell^{({\rm D})} \rangle & =  \frac{1}{4 \pi} \ln \left( \frac{m^2 \ell^2}{4 \pi^2} \right) +   \frac{\gamma}{2 \pi} + \frac{1}{2 \pi}\Re \left( \ee^{\frac{\ii \pi}{\ell} z} \ln \left(1 - \ee^{\frac{\ii 2 \pi}{\ell} z} \right) \right) \nonumber \\
& + \sum_{n = 1}^\infty \left[ \frac{(\mathcal{N}_n^{\rm D})^2}{2\omega_n} \sin^2\left({s_n^{\rm D} z}\right) - \frac{1}{\pi n} \sin^2\left(\frac{\pi}{\ell} (n-1/2) z \right)  \right],
\label{DirichletVacuumBulk}
\end{align}
where the sum appearing on the right-hand side of \eqref{DirichletVacuumBulk} is absolutely  and uniformly convergent in $0<z<\ell$. 

The corresponding boundary state polarization is

\begin{align}
& \langle \Omega_\ell^{({\rm D})} |  ( \hat  \Phi^{\partial}_{\rm ren} )^2 \Omega_\ell^{({\rm D})} \rangle   = \sum_{n = 1}^\infty \frac{(\mathcal{N}^{\rm D}_n)^2}{2\omega^{\rm D}_n} \left[-\beta_1' \sin(\ell s_n^{\rm D}) + \beta_2' s_n^{\rm D} \cos(\ell s_n^{\rm D})\right]^2.
\label{DirichletVacuumBound}
\end{align}

\subsubsection{At temperature $T >0$}

At finite temperature $T = 1/\beta > 0$, we have that the bulk renormalized local state polarization is given by
\begin{align}
\langle(\hat \Phi^{\rm B}_{\rm ren} )^2(t,z) \rangle_{\beta, {\rm D}} & = \langle \Omega_\ell^{{\rm (D)}} | ( \hat \Phi^{{\rm B}}_{\rm ren})^2(t, z) \Omega_\ell^{{\rm (D)}} \rangle + \sum_{n = 1}^\infty  \frac{\left[\psi_n^{\rm (D)} (z) \right]^2}{ \omega_n^{\rm D} \left(\ee^{\beta \omega_n^{\rm D}}-1\right)}, \label{TBulkPolaD}
\end{align}
where the sum on the right-hand side of  \eqref{TBulkPolaD} is  uniformly convergent  in  $z \in (0, \ell),$ 
  and it converges exponentially fast.
The corresponding boundary state polarization is
\begin{align}
\langle( \hat \Phi^{\partial}_{\rm ren} )^2(t) \rangle_{\beta, {\rm D}} & = \langle \Omega_\ell^{{\rm (D)}} | ( \hat  \Phi^{\partial} )^2_{\rm ren}(t) \Omega_\ell^{{\rm (D)}} \rangle + \sum_{n = 1}^\infty  \frac{\left[\psi_n^{\partial \, \rm (D)} \right]^2 }{ \omega_n^{\rm D} \left(\ee^{\beta \omega_n^{\rm D}}-1\right)}, \label{TBoundPolaD}
\end{align}
where the sum on the right-hand side of \eqref{TBoundPolaD}  converges exponentially fast.

\subsection{Local Casimir energy for a Dirichlet boundary condition at $z = 0$}

\subsubsection{At zero temperature}

The bulk local Casimir energy in the ground state is
\begin{align}
 \langle \Omega_\ell^{({\rm D})} |  \hat H^{{\rm B}}_{\rm ren}(t,z) \Omega_\ell^{({\rm D})} \rangle & = \frac{\pi ^2 \beta_2'+6 (2 \gamma -1) \beta_2' \ell^2 m^2+6 \beta_2' \ell^2 m^2 \ln \left(\frac{\ell^2 m^2}{4 \pi ^2}\right)+24 \beta_1' \ell}{48 \pi  \beta_2' \ell^2} \nonumber \\
 & + \frac{m^2}{2 \pi} \Re\left( \ee^{-\ii \frac{\pi}{\ell} z} \ln \left( 1- \ee^{\ii \frac{2 \pi}{\ell} z} \right) \right) + \sum_{n = 1}^\infty \left[  \left( \frac{(\mathcal{N}_n^{\rm D})^2 \omega^{\rm D}_n}{4} - \frac{\pi n }{2 \ell^2} + \frac{\pi }{4  \ell^2}  - \frac{ m^2 }{4 \pi  n} \right) \right. \nonumber \\
&  \left.  - m^2 \left( \frac{(\mathcal{N}_n^{\rm D})^2 }{8 \omega^{\rm D}_n} - \frac{1}{4 \pi n} \right) + \frac{(\mathcal{N}_n^{\rm D})^2 m^2}{4 \omega^{\rm D}_n} \sin^2\left( s_n^{\rm D} z \right) - \frac{m^2}{2 \pi n} \sin^2\left( \frac{\pi}{\ell}(n-1/2) z \right) \right],
\label{HDiriBulk}
\end{align}
where the sum that appears on the right-hand side of  \eqref{HDiriBulk} converges absolutely and uniformly in $z \in (0, \ell)$. 
The corresponding boundary Casimir energy is
\begin{align}
 \langle \Omega_\ell^{({\rm D})} | \hat H^{\partial}_{\rm ren}(t) \Omega_\ell^{({\rm D})} \rangle & =   \sum_{n = 1}^\infty \frac{(\mathcal{N}^{\rm D}_n)^2\left( |\beta_1'| ( \omega_n^{\rm D})^2  - (\text{\rm sign} \, \beta_1' )\beta_1 \right)}{4\omega_n^{\rm D}}  \left[-\beta_1' \sin(\ell s_n^{\rm D}) + \beta_2' s_n^{\rm D} \cos(\ell s_n^{\rm D})\right]^2,
\label{DiriHBound}
\end{align}
where the sum on the right-hand side of  \eqref{DiriHBound} converges absolutely.

\subsubsection{At temperature $T >0$}

At finite temperature $T = 1/\beta > 0$, we have that the bulk local Casimir energy is given by
\begin{align}
\langle& \hat H^{{\rm B}}_{\rm ren}(t, z) \rangle_{\beta, {\rm D}}  = \langle \Omega_\ell^{\rm (D)} | \hat H^{{\rm B}}_{\rm ren}(t, z) \Omega_\ell^{{\rm (D)}} \rangle  + \sum_{n = 1}^\infty \left[\frac{\left( (\omega_n^{\rm D})^2 + m^2\right)}{2 \omega_n^{\rm D}} \frac{\left[\psi_n^{\rm (D)} (z)\right]^2}{\ee^{\beta \omega_n^{\rm D}}-1} + \frac{1}{2 \omega_n^{\rm (D)}} \frac{\left[{\partial_z\psi_n^{\rm (D)}} (z)\right]^2}{\ee^{\beta \omega_n^{\rm D}}-1} \right], \label{HTbulkD}
\end{align}
where the sum on the right-hand side of  \eqref{HTbulkD} is  uniformly convergent in $z \in (0, \ell)$, and it converges exponentially fast. 
The corresponding boundary Casimir energy is
\begin{align}
\langle  \hat H^{\partial}_{\rm ren}(t) \rangle_{\beta, {\rm D}}  = \langle \Omega_\ell^{{\rm (D)}} | \hat H^{\partial}_{\rm ren}(t) \Omega_\ell^{{\rm D}} \rangle + \sum_{n = 1}^\infty \frac{\left( |\beta_1'| (\omega_n^{\rm D})^2 - (\text{\rm sign} \,\beta_1')  \beta_1 \right)}{2 \omega_n^{\rm D}} \frac{\left[\psi_n^{\partial \, \rm (D)}\right]^2 }{\ee^{\beta \omega_n^{\rm D}}-1}, \label{HTboundaryD}
\end{align}
where the sum on the right-hand side of \eqref{HTboundaryD} is absolutely convergent and it converges exponentially fast.

\subsection{Renormalized local state polarization for a Robin boundary condition at $z = 0$}

\subsubsection{At zero temperature}

The bulk renormalized local state polarization in the ground state is
\begin{align}
& \langle \Omega_\ell^{({\rm R})} |  ( \hat \Phi^{{\rm B}}_{\rm ren})^2(t,z) \Omega_\ell^{({\rm R})} \rangle 
 =  \frac{1}{4 \pi} \ln \left( \frac{m^2 \ell^2}{4 \pi^2} \right) +   \frac{\gamma}{2 \pi} - \frac{1}{2 \pi}\Re \left( \ee^{-\frac{\ii \pi}{\ell} z} \ln \left(1 - \ee^{\frac{\ii 2 \pi}{\ell} z} \right) \right) \nonumber \\
& + \sum_{n = 1}^\infty \left\{ \left[ \frac{(\mathcal{N}^{\rm R}_n)^2}{2 \omega_n^{\rm R}} \left(\sin^2 \alpha \cos^2 \left( s_n^{\rm R} z\right) + \frac{\cos^2 \alpha}{(s_n^{\rm R})^2} \sin^2 \left( s_n^{\rm R} z\right) \right)  - \frac{1}{ \pi n} \cos^2 \left( \frac{\pi}{\ell}(n-1) z\right) \right]   \right. \nonumber \\
& \left. - \frac{(\mathcal{N}^{\rm R}_n)^2}{2 \omega_n^{\rm R} s_n^{\rm R}}\sin \alpha \cos \alpha      \sin\left(2 s_n^{\rm R} z \right)\right\},
\label{RobinVacuumBulk}
\end{align}
where the sum the on right-hand side of \eqref{RobinVacuumBulk} converges absolutely and uniformly in $0<z<\ell$.
For the corresponding boundary state polarization we have
\begin{align}
 \langle \Omega_\ell^{({\rm R})} | ( \hat  \Phi^{\partial}_{\rm ren} )^2(t) \Omega_\ell^{({\rm R})} \rangle & = \sum_{n = 1}^\infty \frac{(\mathcal{N}_n^{\rm R})^2}{2\omega_n^{\rm R}} \left[\beta_1' \left(\sin \alpha  \cos (\ell s_n^{\rm R})-\frac{\cos \alpha  \sin (\ell s_n^{\rm R})}{s_n^{\rm R}}\right)  +\beta_2' (\cos \alpha \cos (\ell s_n^{\rm R})+s_n^{\rm R} \sin \alpha  \sin (\ell s_n^{\rm R})) \right]^2.
\label{RobinVacuumBoundary}
\end{align}

\subsubsection{At temperature $T >0$}

At finite temperature $T = 1/\beta > 0$, we have that the bulk renormalized local state polarization is given by
\begin{align}
\langle(\hat \Phi^{\rm B}_{\rm ren} )^2(t,z) \rangle_{\beta, {\rm R} } & = \langle \Omega_\ell^{{\rm (R)}} | ( \hat \Phi^{{\rm B}}_{\rm ren})^2(t, z) \Omega_\ell^{{\rm (R)}} \rangle + \sum_{n = 1}^\infty  \frac{\left[\psi_n^{\rm (R)} (z) \right]^2}{ \omega_n^{\rm R} \left(\ee^{\beta \omega_n^{\rm R}}-1\right)}, \label{TBulkPolaR}
\end{align}
where the sum on the right-hand side of  \eqref{TBulkPolaR} is uniformly convergent in $z \in (0, \ell)$,  and it converges exponentially fast.  
The corresponding boundary state polarization is
\begin{align}
\langle( \hat \Phi^{\partial}_{\rm ren} )^2(t) \rangle_{\beta, {\rm R}} & = \langle \Omega_\ell^{{\rm (R)}} | ( \hat  \Phi^{\partial} )^2_{\rm ren}(t) \Omega_\ell^{{\rm (R)}} \rangle + \sum_{n = 1}^\infty  \frac{\left[\psi_n^{\partial \, \rm (R)} \right]^2 }{ \omega_n^{\rm R} \left(\ee^{\beta \omega_n^{\rm R}}-1\right)}, \label{TBoundPolaR}
\end{align}
where the sum on the right-hand side of \eqref{TBoundPolaR}  converges exponentially fast.

\subsection{Local Casimir energy for a Robin boundary condition at $z = 0$}

\subsubsection{At zero temperature}

The bulk local Casimir energy in the ground  state is
\begin{align}
& \langle  \Omega_\ell^{({\rm R})} |  \hat H^{\rm B}(t,z) \Omega_\ell^{({\rm R})} \rangle = -\frac{\pi }{24 \ell^2}-\frac{\cot \alpha }{2 \pi  \ell}-\frac{\beta_1'}{2 \pi  \beta_2' \ell}+\frac{m^2}{8 \pi }\left[ 1  + \ln \left( \frac{m^2 \ell^2}{4 \pi^2} \right) \right] + \frac{\gamma m^2 }{4 \pi} \nonumber \\
& - \frac{m^2}{2 \pi} \Re \left[\ee^{-\ii \frac{ 2 \pi }{\ell} z} \ln \left(1-\ee^{\ii  \frac{2 \pi}{\ell} z}\right) \right] - m^2 \sum_{n = 1}^\infty \frac{(\mathcal{N}^{\rm R}_n)^2}{2 \omega_n^{\rm R} s_n^{\rm R}}\sin \alpha \cos \alpha      \sin\left(2 s_n^{\rm R} z \right) \nonumber \\
&+\sum_{n = 1}^\infty \left\{  \frac{(\mathcal{N}^{\rm R}_n)^2}{8 \omega_n^{\rm R}} \left((\omega_n^{\rm R})^2 + (s_n^{\rm R})^2 \right)     \left(\sin^2 \alpha + \frac{\cos^2 \alpha}{(s_n^{\rm R})^2}  \right) - \frac{\pi (n-1)}{2 \ell^2}  \right\} \nonumber \\
& + \frac{m^2}{2}\sum_{n = 1}^\infty \left[ \frac{(\mathcal{N}^{\rm R}_n)^2}{4 \omega_n^{\rm R}} \left(\sin^2 \alpha - (s_n^{\rm R})^{-2} \cos^2 \alpha \right) \cos(2 s_n^{\rm R} z)   - \frac{1}{ 2 \pi n} \cos \left( \frac{2(n-1) \pi}{\ell} z \right) \right]  \nonumber \\
& + \frac{m^2}{2} \sum_{n = 1}^\infty  \left[ \frac{(\mathcal{N}^{\rm R}_n)^2}{2 \omega_n^{\rm R}} \left(\sin^2 \alpha \cos^2 \left( s_n^{\rm R} z\right) + \frac{\cos^2 \alpha}{(s_n^{\rm R})^2} \sin^2 \left( s_n^{\rm R} z\right) \right) - \frac{1}{ \pi n} \cos^2 \left( \frac{\pi}{\ell}(n-1) z\right) \right] .
\label{RobinHBulk}
\end{align}
where the sums on the right-hand side of  \eqref{RobinHBulk} converge absolutely and uniformly in $0<z<\ell$.
The corresponding boundary Casimir energy is
\begin{align}
\langle \Omega_\ell^{({\rm R})} | \hat H^\partial_{\rm ren}(t) \Omega_\ell^{({\rm R})} \rangle & = \sum_{n = 1}^\infty \frac{\left[|\beta_1'| (\omega_n^{\rm R})^2 - (\text{\rm sign}    \, \beta_1') \beta_1 \right] (\mathcal{N}_n^{\rm R})^2}{4\omega_n^{\rm R}} \left[\beta_1' \left(\sin \alpha  \cos (\ell s_n^{\rm R})-\frac{\cos \alpha  \sin (\ell s_n^{\rm R})}{s_n^{\rm R}}\right) \right. \nonumber \\
& \left.  +\beta_2' (\cos \alpha \cos (\ell s_n^{\rm R})+s_n^{\rm R} \sin \alpha  \sin (\ell s_n^{\rm R})) \right]^2,
\label{RobinHBoundary}
\end{align}
where the sum on the right-hand side of  \eqref{RobinHBoundary}  is absolutely convergent.

\subsubsection{At temperature $T >0$}

At finite temperature $T = 1/\beta > 0$, we have that the bulk local Casimir energy is given by
\begin{align}
\langle& \hat H^{{\rm B}}_{\rm ren}(t, z) \rangle_{\beta, {\rm R}}  = \langle \Omega_\ell^{\rm (R)} | \hat H^{{\rm B}}_{\rm ren}(t, z) \Omega_\ell^{{\rm (R)}} \rangle  + \sum_{n = 1}^\infty \left[\frac{\left( (\omega_n^{\rm R})^2 + m^2\right)}{2 \omega_n^{\rm R}} \frac{\left[\psi_n^{\rm (R)} (z)\right]^2}{\ee^{\beta \omega_n^{\rm R}}-1} + \frac{1}{2 \omega_n^{\rm R}} \frac{\left[{\partial_z\psi_n^{\rm (R)}} (z)\right]^2}{\ee^{\beta \omega_n^{\rm R}}-1} \right], \label{HTbulkR}
\end{align}
where the sum on the right-hand side of  \eqref{HTbulkR} is   uniformly convergent in $z \in (0, \ell)$, and it converges exponentially fast. 
The corresponding boundary Casimir energy is
\begin{align}
\langle  \hat H^{\partial}_{\rm ren}(t) \rangle_{\beta, {\rm R}}  = \langle \Omega_\ell^{{\rm (R)}} | \hat H^{\partial}_{\rm ren}(t) \Omega_\ell^{{\rm (R)}} \rangle + \sum_{n = 1}^\infty \frac{\left( |\beta_1'| (\omega_n^{\rm R})^2 - (\text{\rm sign}\, \beta_1')   \beta_1 \right)}{2 \omega_n^{\rm R}} \frac{\left[\psi_n^{\partial \, \rm (R)}\right]^2 }{\ee^{\beta \omega_n^{\rm R}}-1} \label{HTboundaryR}
\end{align}
where the sum on the right-hand side of \eqref{HTboundaryR} is absolutely convergent, and it converges exponentially fast.

\end{document}